\documentclass{elsarticle}

\usepackage[utf8]{inputenc}
\usepackage[T1]{fontenc}

\usepackage{amsmath}
\usepackage{amssymb}
\usepackage{amsthm}
\usepackage{bm}
\usepackage{mathrsfs}
\usepackage{csquotes}
\usepackage{tabularx}
\usepackage{enumitem}
\usepackage{hyperref}
\usepackage{algorithm}
\usepackage{algpseudocode}
\usepackage{bbm}
\usepackage{textcomp}
\usepackage{varwidth}
\usepackage{wrapfig}
\usepackage{afterpage}
\usepackage{tikz}
\usetikzlibrary{automata, calc, matrix, shapes, positioning, decorations.pathreplacing}

\usepackage{todonotes}

\newtheorem{example}{Example}

\newtheorem{prop}{Proposition}

\newcommand{\inverse}[1]{\mkern 1.5mu\overline{\mkern-1.5mu#1\mkern-1.5mu}\mkern 1.5mu}

\DeclareFontFamily{U}{mathb}{\hyphenchar\font45}
\DeclareFontShape{U}{mathb}{m}{n}{
  <5> <6> <7> <8> <9> <10> gen * mathb
  <10.95> mathb10 <12> <14.4> <17.28> <20.74> <24.88> mathb12
}{}
\DeclareSymbolFont{mathb}{U}{mathb}{m}{n}
\DeclareFontSubstitution{U}{mathb}{m}{n}
\DeclareMathSymbol{\drsh}{3}{mathb}{"EB}

\newsavebox{\blankboxdisplay}
\savebox{\blankboxdisplay}{\hspace{0.1ex}\tikz[baseline=0.1em]{%
    \node [shape=rectangle, anchor=south, draw, solid, inner sep=0pt, minimum width=1ex, minimum height=0.9em] (char) {};}%
  \hspace{0.1ex}}
\newsavebox{\blankboxtext}
\savebox{\blankboxtext}{\hspace{0.1ex}\tikz[baseline=0.1em]{%
    \node [shape=rectangle, anchor=south, draw, solid, inner sep=0pt, minimum width=1ex, minimum height=0.9em] (char) {};}%
  \hspace{0.1ex}}
\newsavebox{\blankboxscript}
\savebox{\blankboxscript}{\scriptsize\hspace{0.1ex}\tikz[baseline=0.1em]{%
    \node [shape=rectangle, anchor=south, draw, solid, inner sep=0pt, minimum width=1ex, minimum height=0.9em] (char) {};}%
  \hspace{0.1ex}}
\newsavebox{\blankboxscriptscript}
\savebox{\blankboxscriptscript}{\tiny\hspace{0.1ex}\tikz[baseline=0.1em]{%
    \node [shape=rectangle, anchor=south, draw, solid, inner sep=0pt, minimum width=1ex, minimum height=0.9em] (char) {};}%
  \hspace{0.1ex}}
\newcommand{\blank}{\makeatletter%
  \ensuremath{\mathchoice%
    {\text{\usebox\blankboxdisplay}}%
    {\text{\usebox\blankboxtext}}%
    {\text{\usebox\blankboxscript}}%
    {\text{\usebox\blankboxscriptscript}}}%
  \makeatother}

\newlength{\edgelength}
\newcommand{\trans}[4]{%
  \begin{tikzpicture}[auto, shorten >=1pt, >=latex, baseline=(l.base), inner sep=0pt, outer xsep=0.3333em]
    \node (l) {\ensuremath{#1}};%
    \setlength{\edgelength}{\widthof{\scriptsize\ensuremath{#2/#3}}+0.5cm}%
    \node[base right=\edgelength of l] (r) {\ensuremath{#4}};%
    \path[->] (l.mid east) edge node[inner sep=0pt] {\scriptsize\ensuremath{#2/#3}} (r.mid west);%
  \end{tikzpicture}%
}

\newcommand{\transa}[3]{%
  \begin{tikzpicture}[auto, shorten >=1pt, >=latex, baseline=(l.base), inner sep=0pt, outer xsep=0.3333em]
    \node (l) {\ensuremath{#1}};%
    \setlength{\edgelength}{\widthof{\scriptsize\ensuremath{#2}}+0.5cm}%
    \node[base right=\edgelength of l] (r) {\ensuremath{#3}};%
    \path[->] (l.mid east) edge node[inner xsep=0pt, inner ysep=0.2em] {\scriptsize\ensuremath{#2}} (r.mid west);%
  \end{tikzpicture}%
}

\newsavebox{\circledOnebox}
\savebox{\circledOnebox}{\tikz[baseline=(s.base)]{\node[draw, circle, inner sep=0.1mm] (s) {$1$};}}
\newcommand*{\circledOne}{\usebox{\circledOnebox}}

\newsavebox{\circledZerobox}
\savebox{\circledZerobox}{\tikz[baseline=(s.base)]{\node[draw, circle, inner sep=0.1mm, outer sep=0pt] (s) {$0$};}}
\newcommand*{\circledZero}{\usebox{\circledZerobox}}

\DeclareMathOperator{\dom}{dom}
\DeclareMathOperator{\im}{im}
\newcommand{\revbin}{\overleftarrow{\operatorname{bin}}}
\DeclareMathOperator{\idGrp}{\mathbbm{1}}
\newcommand*{\ComplexityClass}[1]{\textsc{#1}}
\newcommand*{\NSPACE}{\ComplexityClass{NSpace}}
\newcommand*{\DSPACE}{\ComplexityClass{DSpace}}
\newcommand*{\PSPACE}{\ComplexityClass{PSpace}}
\newcommand*{\NL}{\ComplexityClass{NL}}

\renewcommand*{\O}{O}

\newcommand*{\id}{\operatorname{id}}

\newcommand*{\sourcekeyword}[1]{\textbf{#1}}

\newcommand{\problem}[3][]{%
  \par\vspace{0.125cm plus 0.1cm minus 0.05cm}\hspace{-0.5\parindent}\begin{tabularx}{\textwidth-0.5\parindent}{l@{\hspace{1ex}}X}%
    \if\relax\detokenize{#1}\relax%
    \else%
      \textnormal{\textbf{Constant:}}&#1\\%
    \fi%
    \textnormal{\textbf{Input:}}&#2\\%
    \textnormal{\textbf{Question:}}&#3\\%
  \end{tabularx}\vspace{0.125cm plus 0.1cm minus 0.05cm}\par%
  }

\begin{document}

  \begin{frontmatter}
    \title{On the Complexity of the Word Problem for\\Automaton Semigroups and Automaton Groups}

    \author[dangeli]{Daniele D'Angeli}
    \ead{dangeli@math.tugraz.at}
    \address[dangeli]{Institut für Diskrete Mathematik,\\
      Technische Universität Graz,\\
      Steyrergasse 30,
      8010 Graz, Austria}

    \author[rodaro]{Emanuele Rodaro\corref{cor}}
    \ead{emanuele.rodaro@polimi.it}
    \address[rodaro]{Dipartimento di Matematica,\\
      Politecnico di Milano,\\
      Piazza Leonardo da Vinci, 32,
      20133 Milano, Italy}

    \author[waechter]{Jan Philipp Wächter}
    \ead{jan-philipp.waechter@fmi.uni-stuttgart.de}
    \address[waechter]{Institut für Formale Methoden der Informatik (FMI),\\
      Universität Stuttgart,
      Universitätsstraße 38,
      70569 Stuttgart, Germany}

    \cortext[cor]{Corresponding author}

    \begin{abstract}
      In this paper, we study the word problem for automaton semigroups and automaton groups from a complexity point of view. As an intermediate concept between automaton semigroups and automaton groups, we introduce automaton-inverse semigroups, which are generated by partial, yet invertible automata. We show that there is an automaton-inverse semigroup and, thus, an automaton semigroup with a $\PSPACE$-complete word problem. We also show that there is an automaton group for which the word problem with a single rational constraint is $\PSPACE$-complete. Additionally, we provide simpler constructions for the uniform word problems of these classes. For the uniform word problem for automaton groups (without rational constraints), we show $\NL$-hardness. Finally, we investigate a question asked by Cain about a better upper bound for the length of a word on which two distinct elements of an automaton semigroup must act differently.

      {\noindent\scriptsize{}A detailed listing of the contributions of this paper can be found at the end of this paper.}
    \end{abstract}


    \end{frontmatter}

  \begin{section}{Introduction}
    Traditionally, algebraic structures have been presented by specifying generators and relations. There is, however, another kind of presentation based on finite Mealy automata. Algorithmically, it has a remarkable advantage: it guarantees decidability of the word problem for the generated groups and semigroups. This stands in contrast to the traditional way of presenting such structures: even if the set of generators and the set of relations are both finite, one can (finitely) present a group with undecidable word problem (a classical result due to Boone and Novikov from the mid 50s). This means in particular that the word problem is not decidable for every group and every semigroup. Thus, not every group or semigroup is an \emph{automaton group} or \emph{automaton semigroup}, i.\,e.\ generated by an automaton. While at first sight this seems to be a severe limitation, automaton groups and semigroups have proven to have deep connections with many areas of mathematics from the theory of profinite groups and complex dynamics to theoretical computer science.

    Many examples of groups with interesting properties are in fact automaton groups. The most notable example is probably Grigorchuk's first example of a group of intermediate (i.\,e.\ super-polynomial but sub-exponential) growth (see e.\,g.\  \cite{grigorchuk2008groups} for an accessible introduction). The existence of such a group answers the classical Milnor problem. This group is also an example of an infinite, finitely generated torsion group (Burnside problem) and an amenable non-elementary amenable group (von Neumann problem). Its discovery led to the development of a new and exiting branch of research in group theory by Grigorchuk and others: the study of finitely presented groups that act (transitively) on each level of an infinite regular rooted tree (see e.\,g.\ \cite{nekrashevych2005self}). This is just a different way of describing automaton groups.

    Decidability of the word problem for automaton groups is most often proven by giving a (deterministic) exponential time algorithm. The idea behind this algorithm is that one can give an exponential upper bound on the level of the tree on which two distinct group elements act differently. For some special sub-classes of automaton groups, there is a better upper bound on how many paths in the tree need to be checked (see e.\,g.\ \cite{bondarenko2012growth} for polynomial-activity automata, \cite{Bondarenko} for Hanoi Tower groups, \cite{nekrashevych2005self} for contracting automaton groups). However, as Steinberg notes in \cite{steinberg2015some}, there is a more straightforward -- at least from a computer science perspective -- nondeterministic algorithm requiring linear space (and, thus, also yielding a deterministic exponential time algorithm). This algorithm leads Steinberg to the question whether there is an automaton group with a \PSPACE-complete word problem.

    In this paper, we give some partial answers to Steinberg's question. We do not only consider automaton groups but also automaton semigroups. This class is a generalization where the generating automata are not required to be invertible. While it is less studied than automaton groups, it has attracted quite some interest recently (for example in \cite{aklmp12, DaRo13, DaRo14, DaRo16, GKP} or in the work of Cain \cite{cain2009automaton}, and Cain and Brough \cite{brough2015automaton, brough2016automaton}). Indeed, Gillibert showed that the finiteness problem for automaton semigroups is undecidable \cite{Gilbert13} while an answer to the corresponding question for automaton groups is yet unknown. It turns out that Steinberg's algorithm is actually an algorithm for automaton semigroups (which, of course, can also be applied for automaton groups); a fact, which we state in \autoref{prop:uniformUpperBound} together with a precise analysis of its complexity. This is interesting since we show \PSPACE-hardness of the word problem for automaton semigroups in the uniform and non-uniform case (\autoref{prop:uniformPSPACEhard} and \autoref{prop:nonuniformPSPACEhardness}).

    Unfortunately, these proofs transfer to automaton groups only in the presence of at least a single rational constraint. To bridge this gap between automaton semigroups and automaton groups, we introduce the concept of automaton-inverse semigroups.\footnote{This concept also appeared for example in \cite{Nekrashevych2006self} and a similar concept was studied by Olijnyk, Sushchansky and Slupik (see \cite{Olijnyk2010inverse} for this line of research).} These semigroups are generated by invertible automata which, however, are not necessarily complete anymore. Since we can prove that the word problem is still \PSPACE-complete in that case (again, uniform and non-uniform), we have the hope that this intermediate concept is useful when transferring other results from automaton semigroups to automaton groups as well.

    For the uniform word problem for automaton groups (without rational constraints), we give a simple construction which shows \NL-hardness (\autoref{prop:groupUniformNLhardness}).

    Additionally, we consider a question asked by Cain in \cite[Open problem~3.6]{cain2009automaton}: is there a better upper bound on the first level of the tree on which two distinct elements of an automaton semigroup act differently than the trivial one? This question is related to the \PSPACE-completeness of the word problem and, therefore, it is probably not surprising that this upper bound cannot be significantly lower: it must be exponential (\autoref{prop:lowerBoundOnWordLength}).
  \end{section}

  \begin{section}{Preliminaries}
    \paragraph{Fundamentals}
    We use $A \sqcup B$ to denote the disjoint union of two sets $A$ and $B$. The set $\mathbb{N}$ is the set of natural numbers including $0$ and $\mathbb{Z}$ is the set of integers. Finally, $\mathbb{R}$ is the set of real numbers.

    An \emph{alphabet} $\Sigma$ is a non-empty finite set of \emph{letters}. By $\Sigma^*$, we denote the set of all finite words over the alphabet $\Sigma$, including the empty word, which we denote by $\varepsilon$. Furthermore, we set $\Sigma^+ = \Sigma^* \setminus \{ \varepsilon \}$.

    We assume the reader to be familiar with basic algebra notions such as that of a semigroup, a monoid, a group and neutral elements (for more information, see for example the book by Howie \cite{howie}).  However, we want to emphasize the difference between a (semigroup) inverse and a group inverse. An element $\inverse{s}$ of a semigroup $S$ is called (semigroup) \emph{inverse} to another element $s \in S$ if $\inverse{s} s \inverse{s} = \inverse{s}$ and $s \inverse{s} s = s$. In contrast, an element $\inverse{m}$ of a monoid $M$ (which may be a group) is a \emph{group inverse} of another element $m \in M$ if $\inverse{m} m = m \inverse{m} = \idGrp$, where $\idGrp$ is the neutral element of $M$. Note that in this case, $\inverse{m}$ is also a (semigroup) inverse of $m$ but that the converse does not hold in general. For instance, consider the monoid $\{ 1, 0 \}$ with the multiplication $1 \cdot 1 = 1$ and $1 \cdot 0 = 0 \cdot 1 = 0 \cdot 0 = 0$ (i.\,e.\ $1$ is the neutral element and $0$ is a zero element). In this monoid, both, $0$ and $1$, are their own (semigroup) inverses, respectively. However, $0$ does not have a group inverse! In this paper, \emph{inverse} always refers to a semigroup inverse. A semigroup $S$ is called an \emph{inverse semigroup} if every element $s \in S$ has a unique inverse $\inverse{s}$ (see e.\,g.\ \cite{howie, Petrich1984} for more information on inverse semigroups).

    Inverse semigroups are closely related to \emph{partial functions} (or \emph{partial maps}). We write $f: A \to_p B$ for a partial function from the set $A$ to the set $B$. By $\dom f$, we denote the \emph{domain} of $f$, i.\,e.\ the set of $a \in A$ such that $f(a)$ is defined. The \emph{image} of $f$, i.\,e.\ the subset $f(\dom f)$ of $B$, is denoted by $\im f$. If $\dom f = A$, then we say $f$ is \emph{total} and write $f: A \to B$. A partial function $\inverse{f} : B \to_p A$ is \emph{inverse} to another partial function $f : A \to_p B$ if $\dom \inverse{f} = \im f$, $\im \inverse{f} = \dom f$ and $f(\inverse{f}(f(a))) = f(a)$ for all $a \in \dom f$ as well as $\inverse{f}(f(\inverse{f}(b))) = \inverse{f}(b)$ for all $b \in \im f$; in slight abuse of terminology, we say that $\inverse{f}$ is a \emph{group inverse} of $f$ if $\dom \inverse{f} = \im f$, $\im \inverse{f} = \dom f$ and $\inverse{f}(f(a)) = a$ for all $a \in \dom f$ as well as $f(\inverse{f}(b)) = b$ for all $b \in \im f$. Note that in the latter case both functions are injective and the group inverse is unique.

    \paragraph{Complexity and Decidability}
    In this paper, we use only basic and well-known parts of computability and complexity theory; for an introduction on the subject see any textbook (e\,g.\ \cite{Pap94}). We also assume the reader to be familiar with Big $\O$ Notation.

    We say a (non-deterministic) Turing Machine is \emph{bounded in space} by a function $s: \mathbb{N} \to \mathbb{N}$ if, on input of a problem instance of length $n$, all computations (whether accepting or not) use at most $s(n)$ positions on any tape.

    Let $f: \mathbb{N} \to \mathbb{R}$ be a function. Then $\NSPACE(f)$ is a set of decision problems: a problem is in $\NSPACE(f)$ if it can be decided by a non-deterministic Turing Machine which is bounded in space by a function $s \in \O(f)$. Analogously, a decision problem is in $\DSPACE(f)$ if it can be decided by a deterministic Turing Machine which is bounded in space by a function $s \in \O(f)$. $\PSPACE$ is the union of $\NSPACE(p)$ for all polynomials $p$, which, by Savitch's Theorem \cite{savitch1970relationships}, coincides with the union of $\DSPACE(p)$ for all polynomials $p$. $\NL$ is the class $\NSPACE(\log)$.

    A problem $P$ is (many-one) \emph{reducible} to a problem $R$ if there is a computable (total) function $f$ that maps an instance $i$ of $P$ to an instance $f(i)$ of $R$ in such a way that $i$ has a positive answer if and only if so has $f(i)$. If $f$ can be computed by a deterministic Turing Machine that is bounded in space by a function $s \in \O(\log)$, then $P$ is \emph{$\log$-space reducible} to $R$. It is well known that $\log$-space reducibility is transitive.

    We say a decision problem $P$ is \emph{hard} (under many-one $\log$-space reduction) for a class $\mathcal{C}$ of decision problems if all problems $R \in \mathcal{C}$ are $\log$-space reducible to $P$. Most of the time, we will show hardness of a problem $P$ for a class $\mathcal{C}$ by showing that a known $\mathcal{C}$-hard problem is $\log$-space reducible to $P$ (using the transitivity of $\log$-space reducibility).

    \paragraph{Automata}\enlargethispage*{\baselineskip}
    A \emph{synchronous finite-state transducer} (without initial or final states), or simply \emph{automaton}, is a tuple $\mathcal{T} = (Q, \Sigma, \Gamma, \delta)$ where $Q$ is a non-empty finite set, $\Sigma$ and $\Gamma$ are alphabets and $\delta$ is a subset of $Q \times \Sigma \times \Gamma \times Q$. The elements of $Q$ are called \emph{states}, $\Sigma$ is the \emph{input alphabet}, $\Gamma$ is the \emph{output alphabet} and an element $(q, a, b, p) \in \delta$, which, from now on, we write as $\trans{q}{a}{b}{p}$, is a \emph{transition} from $q$ to $p$ with \emph{input} $a$ and \emph{output} $b$. The \emph{size} $|\mathcal{T}|$ of an automaton $\mathcal{T}$ is its number of states. Often, automata are represented graphically. In such a graphical representation, we use
    
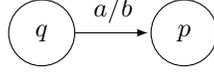
\begin{center}
      \begin{tikzpicture}[auto, shorten >=1pt, >=latex]
        \node[state] (q) {$q$};
        \node[state, right=of q] (p) {$p$};
        \path[->] (q) edge node {$a/b$} (p);
      \end{tikzpicture}
    \end{center}
    to indicate that the automaton has a transition $\trans{q}{a}{b}{p}$.

    An automaton $\mathcal{T} = (Q, \Sigma, \Gamma, \delta)$ is \emph{deterministic} if
    \[
      d_a(q) = \left| \left\{ \trans{q}{a}{b}{p} \mid \trans{q}{a}{b}{p} \in \delta, b \in \Gamma, p \in Q \right\} \right| \leq 1
    \]
    for all $a \in \Sigma$ and $q \in Q$; it is \emph{complete} if $d_a(q) \geq 1$ for all $a \in \Sigma$ and $q \in Q$.

    A deterministic automaton $\mathcal{T} = (Q, \Sigma, \Gamma, \delta)$ induces a partial right action of $\Sigma^*$ on $Q$ defined by $q \cdot a = p$ if $\trans{q}{a}{b}{p} \in \delta$ for some $b \in \Gamma$; the action is total if $\mathcal{T}$ is complete. Additionally, it induces a partial map $\circ: Q \times \Sigma^* \to_p \Gamma^*$ that is defined inductively: let $\circ(q, \varepsilon) = \varepsilon$ and $\circ(q, a) = b$ if $\trans{q}{a}{b}{p} \in \delta$ for some $p \in Q$; for $u \in \Sigma^+$, let $\circ(q, ua)$ be the concatenation of $\circ(q, u)$ and $\circ(q \cdot u, a)$ (if defined). Again, the map is total if $\mathcal{T}$ is complete. For convenience, we write $q \circ u$ instead of $\circ(q, u)$ and consider $q \circ{}$ as the partial map $q \circ{}: \Sigma^* \to_p \Gamma^*$ with $q \circ (u) = q \circ u$. Note that these partial maps are length-preserving (where they are defined).

    We also define \enquote{inverse} versions for these notions; here, we change the role of input and output. For any automaton $\mathcal{T} = (Q, \Sigma, \Gamma, \delta)$, we define its \emph{inverse} automaton $\inverse{\mathcal{T}} = (\inverse{Q}, \Gamma, \Sigma, \inverse{\delta})$ where $\inverse{Q}$ is a disjoint copy of $Q$ and $\inverse{\delta}$ is defined by swapping input and output: $\inverse{\delta} = \{ \trans{\inverse{q}}{b}{a}{\inverse{p}} \mid \trans{q}{a}{b}{p} \in \delta \}$. We say an automaton $\mathcal{T}$ is \emph{inverse-deterministic} if $\inverse{\mathcal{T}}$ is deterministic and $\mathcal{T}$ is \emph{inverse-complete} if $\inverse{\mathcal{T}}$ is complete. Note that, for a deterministic and inverse-deterministic automaton, $q \circ{}$ (where $q$ is a state) and $\inverse{q} \circ{}$ (where $\inverse{q}$ is a state of the inverse automaton) are injective and mutually group inverse.

    An automaton $\mathcal{T} = (Q, \Sigma, \Gamma, \delta)$ is \emph{reversible} if
    \[
      \bar{d}_a(p) = \left| \left\{ \trans{q}{a}{b}{p} \mid \trans{q}{a}{b}{p} \in \delta, b \in \Gamma, q \in Q \right\} \right| \leq 1
    \]
    for all $a \in \Sigma$ and $p \in Q$ (i.\,e.\ if it is co-deterministic with respect to the input). Graphically, $\mathcal{T}$ is reversible if
    
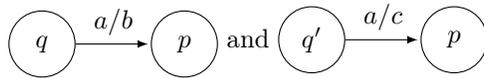
\begin{center}
      \begin{tikzpicture}[auto, shorten >=1pt, >=latex, baseline, every node/.style={anchor=base}]
        \node[state] (q) {$q$};
        \node[state, right=of q] (p) {$p$};
        \path[->] (q) edge node[above] {$a/b$} (p);
      \end{tikzpicture}
      and
      \begin{tikzpicture}[auto, shorten >=1pt, >=latex, baseline, every node/.style={anchor=base}]
        \node[state] (q) {$q'$};
        \node[state, right=of q] (p) {$p$};
        \path[->] (q) edge node[above] {$a/c$} (p);
      \end{tikzpicture}
    \end{center}
    implies $q = q'$ for every $a \in \Sigma$ and every $p \in Q$. It is \emph{bireversible} if, additionally, its inverse is reversible (i.\,e.\ if $\mathcal{T}$ is co-deterministic with respect to the output). {\makeatletter\@endparpenalty=\@medpenalty Graphically, this means that, additionally,
    
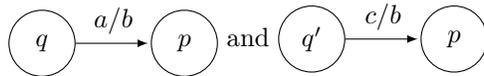
\begin{center}
      \begin{tikzpicture}[auto, shorten >=1pt, >=latex, baseline, every node/.style={anchor=base}]
        \node[state] (q) {$q$};
        \node[state, right=of q] (p) {$p$};
        \path[->] (q) edge node[above] {$a/b$} (p);
      \end{tikzpicture}
      and
      \begin{tikzpicture}[auto, shorten >=1pt, >=latex, baseline, every node/.style={anchor=base}]
        \node[state] (q) {$q'$};
        \node[state, right=of q] (p) {$p$};
        \path[->] (q) edge node[above] {$c/b$} (p);
      \end{tikzpicture}
    \end{center}
    implies $q = q'$ for every $b \in \Gamma$ and every $p \in Q$.}

    Note that reversibility does not imply determinism for non-complete automata. Neither does bireversibility imply inverse-determinism, as the example in \autoref{fig:bireversibleNotInverseDeterministic} demonstrates.
    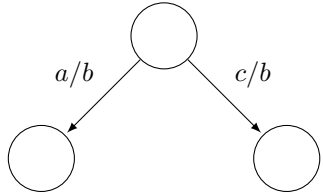
\begin{figure}[h]
      \centering
      \begin{tikzpicture}[auto, shorten >=1pt, >=latex]
        \node[state] (q) {};
        \node[state, below left=of q] (qa) {};
        \node[state, below right=of q] (qc) {};

        \path[->] (q) edge node[swap] {$a/b$} (qa)
                  (q) edge node {$c/b$} (qc);
      \end{tikzpicture}%
      \caption{A bi\-re\-ver\-si\-ble but not in\-verse-de\-ter\-mi\-nis\-tic au\-to\-ma\-ton.}\label{fig:bireversibleNotInverseDeterministic}%
    \end{figure}

    For any two automata $\mathcal{T}_1 = (Q_1, \Sigma_1, \Gamma_1, \delta_1)$ and $\mathcal{T}_2 = (Q_2, \Sigma_2, \Gamma_2, \delta_2)$, one may define their \emph{union} automaton $\mathcal{T}_1 \cup \mathcal{T}_2 = (Q_1 \cup Q_2, \Sigma_1 \cup \Sigma_2, \Gamma_1 \cup \Gamma_2, \delta_1 \cup \delta_2)$.

    \paragraph{Acceptors}
    A \emph{spelling finite-state acceptor}, or simply \emph{acceptor}, is a tuple $\mathcal{A} = (Z, \Sigma, \delta, I,\allowbreak F)$ where $Z$ is a non-empty finite set, $\Sigma$ is an alphabet, $\delta$ is a subset of $Z \times \Sigma \times Z$, $I \subseteq Z$ is a set of \emph{initial} states and $F \subseteq Q$ is a set of \emph{final} states. As with automata, the elements of $Z$ are called \emph{states}, $\Sigma$ is the (input) \emph{alphabet} and an element $(z, a, z') \in \delta$ is an \emph{$a$-transition} from $z$ to $z'$, written as $\transa{z}{a}{z'}$. As a shorthand notation, we write $z \cdot a$ for the set $\{ z' \mid \transa{z}{a}{z'} \in \delta \}$ where $z \in Z$ and $a \in \Sigma$. The \emph{size} $|\mathcal{A}|$ of an acceptor $\mathcal{A}$ is its number of states. We also use the term \emph{$\Sigma$-acceptor} for an acceptor with alphabet $\Sigma$. We represent acceptors similar to automata: we simply omit the output (i.\,e.\ the $/b$ part) and mark final states with double circles (see \autoref{fig:T} for an example of a final state).

    An acceptor $\mathcal{A} = (Z, \Sigma, \delta, I, F)$ is \emph{deterministic} if $|I| = 1$ and
    \[
      d_a(z) = | \{ \transa{z}{a}{z'} \mid \transa{z}{a}{z'} \in \delta, z' \in Z \} | \leq 1
    \]
    for all $a \in \Sigma$ and $z \in Z$; it is \emph{complete} if $d_a(z) \geq 1$ for all $a \in \Sigma$ and $z \in Z$.

    The accepted language $L(\mathcal{A})$ of an acceptor $\mathcal{A} = (Z, \Sigma, \delta, I, F)$ is the set of words $u = a_1 a_2 \dots a_n \in \Sigma^*$ with $a_1, a_2, \dots, a_n \in \Sigma$ for $n \in \mathbb{N}$ such that there is a sequence $(z_0, z_1, \dots, z_n)$ of states with $z_0 \in I$, $\transa{z_{i - 1}}{a_i}{z_i} \in \delta$ for all $i \in \{ 1, 2, \dots, n \}$ and $z_n \in F$.

    \paragraph{Automaton Semigroups and Groups}\enlargethispage*{\baselineskip}
    An \emph{$\mathscr{S}$-automaton} is a deterministic automaton whose input and output alphabet coincide. For brevity, we write $\mathcal{T} = (Q, \Sigma, \delta)$ instead of the more verbose $\mathcal{T} = (Q, \Sigma, \Sigma, \delta)$ in the case of automata with coinciding input and output alphabet. Let $\mathcal{T}$ be an $\mathscr{S}$-automaton with stateset $Q$, then the \emph{semigroup generated by $\mathcal{T}$}, denoted by $\mathscr{S}(\mathcal{T})$, is the closure of $\{ q \circ {} \mid q \in Q \}$ under finite composition of partial maps. A semigroup is an \emph{automaton semigroup}\footnote{Note that the notion of an \emph{automatic} semigroup is not related to that of an automaton semigroup although there is a quite recent attempt to link the two concepts, see \cite{picantin2016automatic}.} if it is isomorphic to the semigroup generated by some $\mathscr{S}$-automaton.

    An inverse-deterministic $\mathscr{S}$-automaton $\mathcal{T} = (Q, \Sigma, \delta)$ is called an \emph{$\inverse{\mathscr{S}}$-au\-to\-ma\-ton}. As an $\inverse{\mathscr{S}}$-automaton $\mathcal{T} = (Q, \Sigma, \delta)$ \emph{generates} an inverse semigroup $\inverse{\mathscr{S}}(\mathcal{T}) = \mathscr{S}(\mathcal{T} \cup \inverse{\mathcal{T}})$; it is the closure of $\{ q \circ{}, \inverse{q} \circ{} \mid q \in Q \}$ under composition of partial maps and, thus, indeed an inverse semigroup. We call a semigroup an \emph{automaton-inverse semigroup} if it is isomorphic to $\inverse{\mathscr{S}}(\mathcal{T})$ for an $\inverse{\mathscr{S}}$-automaton $\mathcal{T}$. Note the difference in definition between \enquote{automaton-inverse semigroups} and \enquote{inverse automaton semigroups} (the latter being automaton semigroups that are inverse semigroups).

    If an $\inverse{\mathscr{S}}$-automaton $\mathcal{T}$ is complete, it is also inverse-complete. In this case, the maps $q \circ{}$ and $\inverse{q} \circ{}$ are total for all states $q$ of $\mathcal{T}$ and the automaton $\mathcal{T}$ is called a \emph{$\mathscr{G}$-automaton}. For a $\mathscr{G}$-automaton $\mathcal{T}$, $\inverse{\mathscr{S}}(\mathcal{T})$ is a group; to emphasize this fact, we use the notation $\mathscr{G}(\mathcal{T})$ for this group. A group that is isomorphic to $\mathscr{G}(\mathcal{T})$ for a $\mathscr{G}$-automaton $\mathcal{T}$ is an \emph{automaton group}.

    To denote an element of an automaton structure (i.\,e.\ an automaton semigroup, au\-to\-ma\-ton-inverse semigroup or automaton group), we omit the $\circ$ and simply write $\bm{q} = q_1 q_2 \dots q_n$ to distinguish it from the partial function $\bm{q} \circ{} = q_1 q_2 \dots q_n \circ{} = q_1 \circ q_2 \circ \dots \circ q_n \circ$ induced by its action on the set of (output) words. This is in line with common algebraic notation. In structures where the inverse is unique, we use the notation $\inverse{{}\mathop{\cdot}{}}$ to denote it; for example, we have $\inverse{q_1 q_2} = \inverse{q_2} \inverse{q_1}$. Note that this is compatible with the notation $\inverse{q}$ for a state of the inverse automaton.

    \begin{example}\label{expl:automatonForZ}
      Let $\mathcal{T}$ be the following automaton.
      
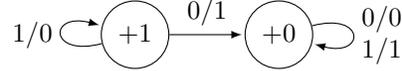
\begin{center}
        \begin{tikzpicture}[auto, shorten >=1pt, >=latex]
          \node[state] (+1) {$+1$};
          \node[state, right=of +1] (+0) {$+0$};

          \path[->] (+1) edge[loop left] node {$1/0$} (+1)
                         edge node {$0/1$} (+0)
                    (+0) edge[loop right] node[align=left] {$0/0$\\$1/1$} (+0)
          ;
        \end{tikzpicture}
      \end{center}
      Obviously, $+0 \circ {}$ is the identity on $\{ 0, 1 \}$ and, thus, $+0$ is a neutral element in the semigroup generated by $\mathcal{T}$. The map $+1 \circ{}$ can be seen as adding one to the least significant bit first\footnote{This is the reverse of the normal binary notation.} binary representation of a natural number. For example, we have $+1 \circ 010 = 110$ and $+1 \circ 110 = 001$. So, $\mathscr{S}(\mathcal{T})$ is the monoid $(\mathbb{N}, +)$ and the group $\mathscr{G}(\mathcal{T})$ is the additive group of integers $(\mathbb{Z}, +)$.
    \end{example}
    \begin{example}
      The following complete automaton $\mathcal{T} = (\Sigma, \Sigma, \delta)$ with $\Sigma = \{ a, b \}$
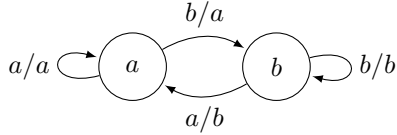
      \begin{center}
        \begin{tikzpicture}[auto, shorten >=1pt, >=latex]
        \node[state] (a) {$a$};
        \node[state, right=of a] (b) {$b$};
        
        \path[->] (a) edge[loop left] node {$a/a$} (a)
                      edge[bend left] node {$b/a$} (b)
                  (b) edge[loop right] node {$b/b$} (b)
                      edge[bend left] node {$a/b$} (a);
        ;
        \end{tikzpicture}
      \end{center}
      generates the free semigroup $\Sigma^+$. This follows from the fact that $\bm{q} \circ u$ is the prefix\footnote{A word $u$ is a prefix of a word $w$ if there is a word $v$ such that $w = uv$ holds.} of length $|\bm{q}|$ of $\bm{q}u$ for all $\bm{q}, u \in \Sigma^*$.\footnote{See also \cite{cain2009automaton}.} Notice that, here, state sequences and words come from the same set!

      If we remove the transition $\trans{b}{a}{b}{a}$, then the resulting automaton is not complete anymore and $b \circ{}\!$ is only defined on words from $\{ b \}^+$. As a consequence, we have that $ba \circ{}\!$ is undefined on any input word (except for the empty word, of course). This turns $ba$ into a zero of the semigroup generated by the modified automaton. Therefore, all non-zero elements are of the form $a^n b^m$ with $n, m \in \mathbb{N}$. Since we have $bb \circ{}\! = b \circ{}\!$, $b$ is idempotent in this semigroup, which leaves us with the elements $ba$, $a^n$ and $a^n b$ for $n \in \mathbb{N}$ that are pairwise distinct.
    \end{example}

    Usually, an automaton semigroup is defined as the semigroup generated by a \emph{complete} $\mathscr{S}$-automaton. We will call such a semigroup a \emph{complete automaton semigroup}. Including semigroups generated by non-complete $\mathscr{S}$-automata is a natural extension of this notion (which allows for defining automaton-inverse semigroups as we did above). It is clear that any complete automaton semigroup is an automaton semigroup (as defined above). On the other hand, if $S$ is an automaton semigroup, then $S^0$, i.\,e.\ the semigroup resulting from adjoining a zero to $S$, is a complete automaton semigroup.

    \begin{prop}\label{prop:SAutSGImpliesS0CompAutSG}
      If $S$ is an automaton semigroup, then $S^0$ is a complete automaton semigroup.
    \end{prop}
    \begin{proof}
      To see this, let $\mathcal{T} = (Q, \Sigma, \delta)$ be an $\mathscr{S}$-automaton generating $S$. We extend $\mathcal{T}$ into a complete automaton $\hat{\mathcal{T}} = (\hat{Q} \sqcup \{ 0 \}, \Sigma \sqcup \{ \bot \}, \hat{\delta})$ where $\hat{Q}$ is a disjoint copy of $Q$, $0$ is a new (sink) state and $\bot$ is a new symbol. The new automaton $\hat{\mathcal{T}}$ has all of $\mathcal{T}$'s transitions (i.\,e.\ we have $\trans{\hat{q}}{a}{b}{\hat{p}} \in \hat{\delta}$ if $\trans{q}{a}{b}{p} \in \delta$) and, additionally, new transitions $\trans{\hat{q}}{a}{\bot}{0}$ whenever there is no transition with input $a$ from state $q$ in $\delta$. We also add the transitions $\trans{0}{a}{\bot}{0}$ and $\trans{\hat{q}}{\bot}{\bot}{0}$ to $\hat{\delta}$ for all $a \in \Sigma \sqcup \{ \bot \}$ and all $q \in Q$. One may verify that $\hat{\mathcal{T}}$ is complete and that $0$ is a zero of its generated semigroup, i.\,e.\ we have $0\hat{q} = 0 = \hat{q}0 = 00$ for all $q \in Q$ in the generated semigroup.

      More importantly, we have $\boldsymbol{q} = \boldsymbol{p}$ in $\mathscr{S}(\mathcal{T})$ if and only if $\bm{\hat{q}} = \bm{\hat{p}}$ in $\mathscr{S}(\hat{\mathcal{T}})$ where $\bm{q} = q_n q_{n - 1} \dots q_1$ and $\bm{p} = p_m p_{m - 1} \dots p_1$ for arbitrary states $q_1, q_2, \dots, q_n,\allowbreak p_1, p_2, \dots, p_m \in Q$ and where $\bm{\hat{q}} = \hat{q}_n \hat{q}_{n - 1} \dots \hat{q}_1$ and $\bm{\hat{p}} = \hat{p}_m \hat{p}_{m - 1} \dots \hat{p}_1$ are the corresponding state sequences for $\hat{\mathcal{T}}$. Note that, if $\bm{q} \circ u$ is defined and equal to $v \in \Sigma^*$ for a $u \in \Sigma^*$, then, by construction, so is $\bm{\hat{q}} \circ u$. Now, suppose we have $\bm{q} \neq \bm{p}$ in $\mathscr{S}(\mathcal{T})$. The first case is that there is a $u \in \Sigma^*$ on which the action of one is defined (say $\bm{q}$) while the other's is not. As we have just discussed, we have $\bm{q} \circ u = \bm{\hat{q}} \circ u \in \Sigma^*$. Furthermore, we have $\bm{\hat{p}} \circ u = v' \bot^k$ for some $k > 0$ and $v \in \Sigma^*$ because $\bm{p} \circ u$ is undefined. This shows $\bm{\hat{q}} \neq \bm{\hat{p}}$ in $\mathscr{S}(\hat{\mathcal{T}})$. The second case is that there is a $u \in \Sigma^*$ such that both, $\bm{q} \circ u$ and $\bm{p} \circ u$, are defined but we have $\bm{p} \circ u = v_p \neq v_q = \bm{q} \circ u$. However, this implies $\bm{\hat{p}} \circ u = v_p \neq v_q = \bm{\hat{q}} \circ u$ and we have $\bm{\hat{p}} \neq \bm{\hat{q}}$ in $\mathscr{S}(\hat{\mathcal{T}})$.

      For the other direction, note that, if $\bm{\hat{q}} \circ u = v \in \Sigma^*$ for a $u \in \Sigma^*$ (i.\,e.\ there is no $\bot$ in $v$), then, by construction, $\bm{q} \circ u$ is defined and equal to $v$. Suppose we have $\bm{\hat{q}} \neq \bm{\hat{p}}$ in $\mathscr{S}(\hat{\mathcal{T}})$. Then, there is a word $ua$ with $u \in ( \Sigma \sqcup \{ \bot \})^*$ and $a \in \Sigma \sqcup \{ \bot \}$ such that $\bm{\hat{q}} \circ u = \bm{\hat{p}} \circ u = v$ but $\bm{\hat{q}} \circ ua = vb \neq vc = \bm{\hat{p}} \circ ua$ for $b, c \in \Sigma \sqcup \{ \bot \}$. Note that $\hat{\mathcal{T}}$ is in state $0$ after reading $\bot$ and that, therefore, $u$ cannot contain $\bot$. Similarly, $\hat{\mathcal{T}}$ will be in state $0$ after an output of $\bot$. Thus, $v$ cannot contain $\bot$ either and we have $\bm{q} \circ u = \bm{p} \circ u = v$. If $b, c \in \Sigma$, then, by the same argumentation, we have $\bm{q} \circ u = vb \neq vc = \bm{p} \circ u$. If (without loss of generality) $b = \bot$ and $c \in \Sigma$, then $\bm{q} \circ ua$ is undefined while $\bm{p} \circ ua = vc$ is not. In either case, we have $\bm{q} \neq \bm{p}$ in $\mathscr{S}(\mathcal{T})$.

      This shows that $\iota: S^0 \to \mathscr{S}(\hat{\mathcal{T}})$ given by $\iota(0) = 0$ and $\iota(q) = \hat{q}$ for every $q \in Q$ is a well-defined isomorphism and that, thus, $S^0$ is a complete automaton semigroup.
    \end{proof}

    We want to point out that the obvious question whether there is an automaton semigroup that is not a complete automaton semigroup is closely related to a question asked by Cain in \cite[Open problem 5.3]{cain2009automaton}: is there a semigroup $S$ such that $S$ is not a complete automaton semigroup but $S^0$ is?
  \end{section}

  \begin{section}{Complexity of the Word Problem}
    Before we start discussing the word problem's complexity, we introduce consistent terminology to distinguish some of its variants more clearly. The variables appearing in the definitions of the problems are used in the same meaning in the statements and proofs of the propositions below.
    \paragraph{Word Problems}
    The \emph{uniform word problem for automaton semigroups with rational constraints} is the decision problem:
    \problem
      { an $\mathscr{S}$-automaton $\mathcal{T} = (Q, \Sigma, \delta)$,\newline
        $q_1, q_2, \dots, q_n, p_1, p_2, \dots, p_m \in Q$,\newline
        $\Sigma$-acceptors $\mathcal{A}_1, \mathcal{A}_2 \dots, \mathcal{A}_r$ of size at most $R$}
      {is $q_n q_{n - 1} \dots q_1 \circ u = p_m p_{m - 1} \dots p_1 \circ u$ for all $u \in \bigcap_{k = 1}^r L(\mathcal{A}_k)$?}
    \noindent{}The \emph{uniform word problem for automaton-inverse semigroups with rational constraints} is the decision problem:
    \problem
      { an $\inverse{\mathscr{S}}$-automaton $\mathcal{T} = (Q, \Sigma, \delta)$,\newline
        $q_1, q_2, \dots, q_n, p_1, p_2, \dots, p_m \in Q \cup \inverse{Q}$ ($\inverse{Q}$ are the states of  $\inverse{\mathcal{T}}$),\newline
        $\Sigma$-acceptors $\mathcal{A}_1, \mathcal{A}_2 \dots, \mathcal{A}_r$ of size at most $R$}
      {is $q_n q_{n - 1} \dots q_1 \circ u = p_m p_{m - 1} \dots p_1 \circ u$ for all $u \in \bigcap_{k = 1}^r L(\mathcal{A}_k)$?}
    \noindent{}A variation of this problem is the \emph{uniform word problem for automaton groups with rational constraints} where the input $\inverse{\mathscr{S}}$-automaton is known to be a $\mathscr{G}$-automaton.

    Additionally, all problems also have versions without rational constraints (then simply called the \emph{uniform word problem for \dots}). Here, the acceptors are omitted from the input and the question is: is $q_n q_{n - 1} \dots q_1 \circ u = p_m p_{m - 1} \dots p_1 \circ u$ for all $u \in \Sigma^*$?

    Besides the uniform versions, we also consider the word problems for specific automaton structures. The \emph{word problem with rational constraints for an automaton semigroup} is the decision problem:
    \problem
      [ an $\mathscr{S}$-automaton $\mathcal{T} = (Q, \Sigma, \delta)$ ]
      { $q_1, q_2, \dots, q_n, p_1, p_2, \dots, p_m \in Q$,\newline
        $\Sigma$-acceptors $\mathcal{A}_1, \mathcal{A}_2 \dots, \mathcal{A}_r$ of size at most $R$}
      {is $q_n q_{n - 1} \dots q_1 \circ u = p_m p_{m - 1} \dots p_1 \circ u$ for all $u \in \bigcap_{k = 1}^r L(\mathcal{A}_k)$?}
    \noindent{}The \emph{word problem with rational constraints for an automaton-inverse semigroup} is the decision problem:
    \problem
      [ an $\inverse{\mathscr{S}}$-automaton $\mathcal{T} = (Q, \Sigma, \delta)$ and its inverse $\inverse{\mathcal{T}} = (\inverse{Q}, \Sigma, \inverse{\delta})$ ]
      { $q_1, q_2, \dots, q_n, p_1, p_2, \dots, p_m \in Q \cup \inverse{Q}$,\newline
        $\Sigma$-acceptors $\mathcal{A}_1, \mathcal{A}_2 \dots, \mathcal{A}_r$ of size at most $R$}
      {is $q_n q_{n - 1} \dots q_1 \circ u = p_m p_{m - 1} \dots p_1 \circ u$ for all $u \in \bigcap_{k = 1}^r L(\mathcal{A}_k)$?}
    \noindent{}For the \emph{word problem with rational constraints for an automaton group}, we additionally have that $\mathcal{T}$ is a $\mathscr{G}$-automaton. Again, we also have versions without rational constraints, simply called the \emph{word problem for an automaton semigroup/inverse semigroup/group}, where the acceptors are omitted from the input and the question is changed in the same way as for the uniform versions.

    \begin{subsection}{Upper Bounds}
      An upper bound for a problem's complexity is usually given in the form of an algorithm and this is what we do in this subsection. Steinberg \cite{steinberg2015some} noted that the word problems (without rational constraints) for automaton groups are decidable in $\NSPACE(n)$. The algorithm he presented is quite straightforward and similar to the one we present here. However, we adapt it for automaton semigroups and automaton-inverse semigroups and extend it to cover rational constraints.
      \begin{prop}\label{prop:uniformUpperBound}
        The following problems can be decided by a non-deterministic Turing Machine in space $\O((n + m) \log \left| \mathcal{T} \right| + r \log R)$ and are, thus, in non-deterministic linear space:
        \begin{itemize}[noitemsep]
          \item the uniform word problem for automaton semigroups with rational constraints,
          \item the uniform word problem for automaton-inverse semigroups with rational constraints and
          \item the uniform word problem for automaton groups with rational constraints.
        \end{itemize}
        Furthermore, the versions without rational constraints can be decided in non-deterministic space $\O((n + m) \log \left| \mathcal{T} \right|)$ and are, thus, also in non-deterministic linear space.
      \end{prop}
      \begin{proof}
        We give an algorithm to decide the uniform word problem for automaton semigroups with rational constraints (as stated above) since this is the most general case. In fact, we give an algorithm for the problem's complement. As non-deterministic space classes are closed under complementation \cite{immerman1988nondeterministic, szelepcsenyi1988method}, this proves the proposition.

        For the input $\Sigma$-acceptors $\mathcal{A}_1, \mathcal{A}_2, \dots, \mathcal{A}_r$, let $\mathcal{A}_k = (Z_k, \Sigma, \delta_k, I_k, F_k)$. The algorithm uses the following variables.
        {\makeatletter
          \@beginparpenalty=0%
         \makeatother%
        \begin{itemize}\nopagebreak
          \item For each input $\Sigma$-acceptor $\mathcal{A}_k$, we need a variable $z_k$ that stores a state from $Z_k$. The initial value of these variables may be undefined.

          Storing a single state requires space $\O(\log R)$; storing all of them, thus, requires space $\O(r \log R)$.

          \item We also need $n + m$ variables to store states of the input automaton $\mathcal{T}$. We re-use the names $q_1, q_2, \dots, q_n, p_1, p_2, \dots, p_m$ for this, their initial values shall be the respective input values.

          Storing a single of these states requires $\O(\log \left| \mathcal{T} \right|)$; storing all of them, thus, requires space $\O((n + m) \log \left| \mathcal{T} \right|)$.

          \item Finally, we need four variables $a$, $a'$, $b$ and $b'$ to store letters from $\Sigma$. Their space requirement is constant\footnote{or in $\O(\log |\Sigma|)$, depending on one's point of view. Either way, the algorithm is in non-deterministic linear space in the end.}.
        \end{itemize}}

        \afterpage{%
          \begin{algorithm}[t]
            \caption{Deciding the uniform word problem for automaton semigroups with rational constraints in $\NSPACE(n)$}\label{alg:nlognAlg}
            \begin{algorithmic}
              \ForAll{$i \in \{ 1, 2, \dots, r \}$}
                \State $z_i \gets \Call{guess}{I_i}$ \Comment{Non-deterministically choose an initial state for each\\\hfill acceptor.}
              \EndFor
              \While{\sourcekeyword{true}}
                \State $a \gets \Call{guess}{\Sigma}$ \Comment{Guess the witness's next letter $a_0 \in \Sigma$.}
                \State $b \gets a$
                \ForAll{$k \in \{ 1, 2, \dots, r \}$}
                  \State $z_k \gets \Call{guess}{z_k \cdot a}$ \Comment{Guess an accepting run in each acceptor.}
                \EndFor
                \ForAll{$i \in \{ 1, 2, \dots, n \}$}
                  \State $a' \gets q_i \circ a$
                  \State $q_i \gets q_i \cdot a$
                  \State $a \gets a'$
                \EndFor \Comment{Here, we have $a = q_n q_{n - 1} \dots q_1 \circ a_0$.}
                \ForAll{$j \in \{ 1, 2, \dots, m \}$}
                  \State $b' \gets p_j \circ b$
                  \State $p_j \gets p_j \cdot b$
                  \State $b \gets b'$
                \EndFor \Comment{Here, we have $b = p_m p_{m - 1} \dots p_1 \circ a_0$.}
                \If{$a \neq b$ \sourcekeyword{and} $\forall k \in \{ 1, 2, \dots, r \}: z_k \in F_k$}
                  \State \Return{\enquote{$\neq$}} \Comment{We have found a witness.}
                \EndIf
              \EndWhile
            \end{algorithmic}
          \end{algorithm}%
        }

        The general idea of the algorithm is now as follows. We guess a word $u \in \Sigma^*$ letter by letter (we only store the last guessed letter of the word). The guessed word should be a witness to prove the inequality of $q_n q_{n - 1} \dots q_1 \circ{}\!$ and $p_m p_{m - 1} \dots p_1 \circ{}\!$ (i.\,e.\ we want to have $q_n q_{n - 1} \dots q_1 \circ u \neq p_m p_{m - 1} \dots p_1 \circ u$) and it should be in $\bigcap_{i = 1}^r L(\mathcal{A}_i)$. To test the latter, we store, for all acceptors, one of the states in which the acceptor can be after reading the current prefix of $u$; then, we can test whether that state is final, in which case the prefix is in the acceptor's language. To test the former, we store the $n + m$ states of the automaton $\mathcal{T}$ in which we are after reading the current prefix of $u$. After guessing the next letter $a \in \Sigma$, we compute successively $a_i = q_i \circ a_{i - 1}$ and $q_i \cdot a_{i - 1}$ as the new value of $q_i$ where $a_0 = a$ (without storing the intermediate results). We do the same for $p_1, p_2, \dots, p_m$ and obtain the letter $b_m$ in the end. Then, we compare $a_n$ and $b_m$. If they are unequal (and the acceptors accept), then we have found the witness and can return \enquote{$\neq$}. The algorithm's pseudo-code may be found in \autoref{alg:nlognAlg}. Its correctness is easy to see: if we have inequality, then the algorithm guesses a witness, and, if it guesses a witness, then we have inequality.

        Handling automaton-inverse semigroups and groups only requires storing an additional bit for each state in the automaton that indicates whether we are in a state $q$ or in a state $\inverse{q}$; the rest of the algorithm is structurally equal. The algorithms for the problems without rational constraints do not require the additional book-keeping for the acceptors and are, thus, even simpler (one can simply remove the corresponding commands from the algorithm).
      \end{proof}

      There is an interesting observation concerning the last proposition and its proof: to show that the uniform word problem for automaton groups is in non-deterministic linear space, we nowhere exploit group properties. In fact, we only give an algorithm for automaton semigroups that can also be applied to groups. Interestingly, we will see later (in \autoref{prop:uniformPSPACEhard}) that the uniform word problem for automaton semigroups (and automaton-inverse semigroups) is $\PSPACE$-hard. Therefore, the presented simple algorithm -- in the sense of complexity -- is as efficient as possible. However, it is also possible that there is a more efficient algorithm for the uniform word problem for automaton groups without rational constraints (making use of their group structure).

      The algorithm from the previous proposition's proof can be adapted to decide the word problem for a fixed automaton structure. Here, the automaton $\mathcal{T}$ is not considered part of the input. Therefore, we can store a single state of the automaton in constant space and, as a simple consequence, get the following proposition.
      \begin{prop}\label{prop:nonuniformUpperBound}
        These problems can be decided by a non-deterministic Turing Machine in space $\O(n + m + r \log R)$ and are, thus, in non-deterministic linear space:
        \begin{itemize}[noitemsep]
          \item the word problem with rational constraints for any automaton semigroup,
          \item the word problem with rational constraints for any automaton-inverse semigroup and
          \item the word problem with rational constraints for any automaton group.
        \end{itemize}
        Furthermore, the versions without rational constraints can be decided in non-deterministic space $\O(n + m)$ and are, thus, also in non-deterministic linear space.
      \end{prop}
    \end{subsection}

    \begin{subsection}{Lower Bounds}
      In the last section, we gave a \emph{non-deterministic} algorithm to decide many variations of the word problem for automaton structures. As the algorithms only require linear space, they admit \emph{deterministic} exponential time algorithms (see e.\,g.\ \cite{Pap94}). While it is a less precise approach from a complexity viewpoint, such deterministic algorithms are often stated directly by using the fact that one can give an upper bound for the length of a word on which two distinct elements of an automaton structure act differently. This path was also taken by Cain in \cite{cain2009automaton}, where he asks whether there is a better upper bound on the word length than $|Q|^{n + m}$ in the case of automaton semigroups (the notation used here is the same as in the definitions of the word problems above). In the next proposition, we handle this question by giving lower bounds on the length.
      \begin{prop}\label{prop:lowerBoundOnWordLength}
        Let $\mathcal{D} = (Q, \{ a, b \}, \delta)$ be the $\mathscr{S}$-automaton\footnote{We want to remark that the presented $\mathscr{S}$-automaton is the so-called \emph{dual} of the adding machine depicted in \autoref{expl:automatonForZ}. For the construction of the dual, states and letters change roles, just like the operations $\circ$ and $\cdot$.}
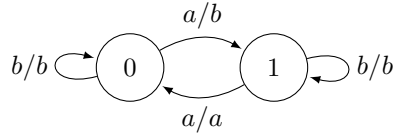
        \begin{center}
          \begin{tikzpicture}[auto, shorten >=1pt, >=latex, baseline]
            \node[state] (0) {$0$};
            \node[state, right=of 0] (1) {$1$};

            \path[->] (0) edge[loop left] node {$b/b$} (0)
                          edge[bend left] node {$a/b$} (1)
                      (1) edge[loop right] node {$b/b$} (1)
                          edge[bend left] node {$a/a$} (0);
          \end{tikzpicture}
        \end{center}
        and let $n > 0$. Then, the elements $0^n$ and $0^{n - 1}$ of $\mathscr{S}(\mathcal{D})$ are distinct but $0^n \circ u = 0^{n - 1} \circ u$ holds for all $u \in \{ a, b \}^*$ of length smaller than $2^{n - 1} = |Q|^{\left\lfloor \frac{|0^n| + |0^{n - 1}|}{2} \right\rfloor}$.

        Let $\mathcal{D}' = (Q', \{ a, b \}, \delta')$ be the $\mathscr{S}$-automaton obtained from $\mathcal{D}$ by adding the state
        
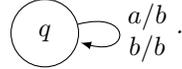
\begin{center}
          \begin{tikzpicture}[auto, shorten >=1pt, >=latex, baseline]
            \node[state] (q) {$q$};
            \path[->] (q) edge[loop right] node[align=center] {$a/b$\\$b/b$} (q);
          \end{tikzpicture}.
        \end{center}
        Then, $0^{n - 1}$ and $q$ are distinct elements of $\mathscr{S}(\mathcal{D}')$ but $0^{n - 1} \circ u = q \circ u$ holds for all $u \in \{ a, b \}^*$ of length smaller than $2^{n - 1} = (|Q'| - 1)^{|0^{n - 1}| + |q| - 1}$.
      \end{prop}
      \begin{proof}
        For a natural number $i \in \mathbb{N}$, let $\revbin_k(i)$ denote $i$'s binary representation of length $k$ in reverse (i.\,e.\ with least significant bit first). One may verify that, for $0 \leq i < 2^k - 1$, we have
        \begin{align*}
          \revbin_k(i) \circ a &= b \text{,}& \revbin_k(i) \cdot a &= \revbin_k(i + 1) \text{,}\\
          \revbin_k(i) \circ b &= b \text{ and}& \revbin_k(i) \cdot b &= \revbin_k(i) \text{.}
        \end{align*}
        In other words: an input of $a$ is adding one to the binary representation and outputs $b$ while an input of $b$ does not change anything and outputs $b$ again. This shows $0^n \circ u = b^{|u|} = 0^{n - 1} \circ u = q \circ u$ for all $u \in \{ a, b \}^*$ of length up to $2^{n - 1} - 1$. Additionally, we have $0^{n - 1} \cdot a^{2^{n - 1} - 1} = \revbin_{n - 1}(2^{n - 1} - 1) = 1^{n - 1}$ and $0^{n} \cdot a^{2^{n - 1} - 1} = \revbin_n(2^{n - 1} - 1) = 0 1^{n - 1}$, which shows $0^{n - 1} \circ a^{2^{n - 1}} = a^{2^{n - 1} - 1} b \neq a^{2^{n - 1}} = 0^n \circ b^{2^{n - 1}} = q \circ a^{2^{n - 1}}$.
      \end{proof}

      In some sense, the previous proposition gives a lower bound for the running time of the often stated algorithm for the word problem that enumerates all words up to a certain length and tests whether the input elements act differently on them. In the next proposition, we show a general, algorithm-agnostic lower bound by showing $\PSPACE$-hardness for some versions of the uniform problem. Please note that these results also follow from \autoref{prop:nonuniformPSPACEhardness} below. However, the encoding technique used in the next proof is much simpler and more direct. Therefore, it is interesting in its own right and we feel that it deserves to be mentioned explicitly.
      \begin{prop}\label{prop:uniformPSPACEhard}
        These problems are $\PSPACE$-hard:
        \begin{itemize}[noitemsep]
          \item the uniform word problem for automaton semigroups (with or without rational constraints),
          \item the uniform word problem for automaton-inverse semigroups (with or without rational constraints),
          \item the uniform word problem for automaton groups with a single rational constraints.
        \end{itemize}
      \end{prop}
      \begin{proof}
        From a classical result of Kozen \cite{kozen77lower} follows $\PSPACE$-hardness of the following problem.
        \problem
          [$\Sigma = \{ 0, 1 \}$]
          { $r \in \mathbb{N}$,
            complete and deterministic $\Sigma$-acceptors $\mathcal{A}_1, \mathcal{A}_2, \dots, \mathcal{A}_r$ of size at most $R$}
          {is $\bigcap_{k = 1}^r L(\mathcal{A}_k)$ empty?}
        We describe how to map an instance $\mathcal{A}_1, \mathcal{A}_2, \dots, \mathcal{A}_r$ of the previous problem to an instance of the uniform word problem for automaton-inverse semigroups (without rational constraints). Let $\mathcal{A}_k = (Z_k, \Sigma, \delta_k, \{ z_{0, k} \}, F_k)$ for $k = 1, 2, \dots, r$. Without loss of generality, we assume all $Z_k$ to be disjoint. For each $k \in \{ 1, 2, \dots, r \}$, let $\mathcal{T}_k'$ be the automaton over the alphabet $\Gamma = \{ 0, 1, \# \}$ given by the graphical representation in \autoref{fig:Tkp} (without the dashed transitions).
        \begin{figure}[h]\caption{a schematic representation of $\mathcal{T}_k'$ and the transitions to make it complete}\label{fig:Tkp}
          \begin{center}
            \resizebox{\linewidth}{!}{
              \begin{tikzpicture}[auto, shorten >=1pt, >=latex]
                \node[state] (0) {$0$};
                \node[state, right=of 0] (1) {$1$};
                \node[state, right=of 1] (2) {$2$};
                \node[right=of 2] (d) {$\dots$};
                \node[state, right=0cm of d] (k-1) {$k - 1$};
                \node[state, right=of k-1] (k) {$k$};
                \node[state, right=of k] (k+1) {$k + 1$};
                \node[state, right=of k+1] (k+2) {$k + 2$};
                \node[right=of k+2] (d2) {$\dots$};
                \node[state, right=0cm of d2] (r) {$r$};
                \node[state, below=of r] (r+1) {};
                \node[state, left=of r+1] (r+2) {};

                \path[->] (0) edge node[below, align=left] {$0/0$\\$1/1$} (1)
                          (1) edge node[below, align=left] {$0/0$\\$1/1$} (2)
                          (2) edge node[below, align=left] {$0/0$\\$1/1$} (d)
                          (k-1) edge node[below, align=left] {$1/0$} (k)
                          (k) edge node[below, align=left] {$0/0$\\$1/1$} (k+1)
                          (k+1) edge node[below, align=left] {$0/0$\\$1/1$} (k+2)
                          (k+2) edge node[below, align=left] {$0/0$\\$1/1$} (d2)
                          (r) edge node[left] {$\#/\#$} (r+1)
                          (r+1) edge node[above] {$1/1$} (r+2)

                          (0) edge[bend left, dashed] node {$\#/\#$} (1)
                          (1) edge[bend left, dashed] node {$\#/\#$} (2)
                          (2) edge[bend left, dashed] node[above] {$\#/\#$} (d)
                          (k-1) edge[bend left, dashed] node[above, align=center] {$0/1$\\$\#/\#$} (k)
                          (k) edge[bend left, dashed] node[above] {$\#/\#$} (k+1)
                          (k+1) edge[bend left, dashed] node[above] {$\#/\#$} (k+2)
                          (k+2) edge[bend left, dashed] node[above] {$\#/\#$} (d2)
                          (r) edge[bend left, dashed] node[right, align=center] {$0/0$\\$1/1$} (r+1)
                          (r+1) edge[bend left, dashed] node[below, align=center] {$0/0$\\$\#/\#$} (r+2)
                          (r+2) edge[loop left, dashed] node[align=center] {$0/0$\\$1/1$\\$\#/\#$} (r+2)
                ;
              \end{tikzpicture}}
          \end{center}
        \end{figure}
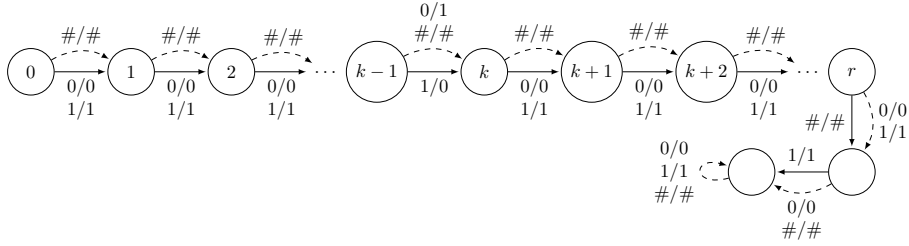
        Note that $\mathcal{T}_k'$ is an $\inverse{\mathscr{S}}$-automaton. The idea of the construction is to ensure that an input word is (a prefix of a word) of the form $\{ 0, 1 \}^{k - 1} 1 \{ 0, 1 \}^{r - k} \# 1$, that is a block of length $r$ consisting of $0$s and $1$s, where the $k^{\textnormal{th}}$ letter is $1$, followed by $\# 1$. This $k^{\textnormal{th}}$ letter is switched from $1$ to $0$ while all other letters are left unchanged.

        Besides $\mathcal{T}_k'$, we also define $\mathcal{T}_k''$, a disjoint copy of $\mathcal{T}_k'$ where the transition $\trans{(k - 1)}{1}{0}{k}$ is changed to $\trans{(k - 1)}{1}{1}{k}$, meaning that the $1$ at position $k$ is checked but not changed. Furthermore, we define $\mathcal{T}_k$ to be the automaton obtained from the acceptor $\mathcal{A}_k$ by changing all transitions $\transa{z}{a}{z'}$ into transitions $\trans{z}{a}{a}{z'}$, that is we add to each transition an output which is equal to the input. Note that the resulting automata are $\inverse{\mathscr{S}}$-automata; in fact, they are $\mathscr{G}$-automata.

        Finally, we define $\mathcal{T}'$ to be given by the graphical representation in \autoref{fig:Tp} (without the dashed transitions). Clearly, it is an $\inverse{\mathscr{S}}$-automaton.
        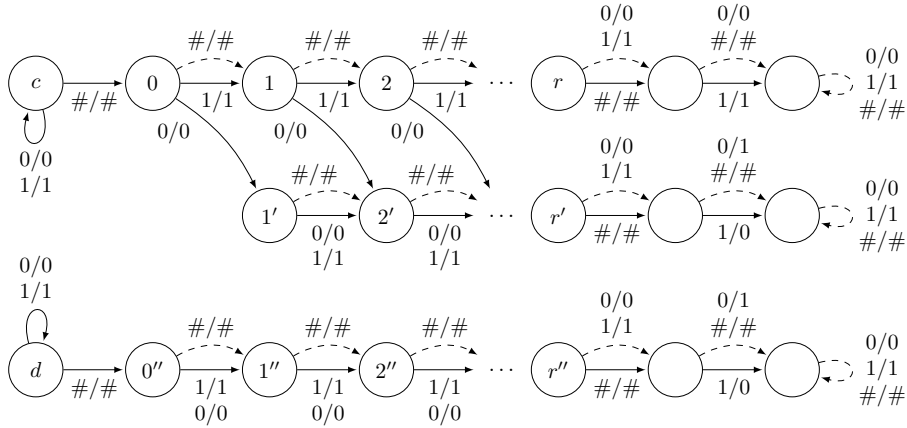
\begin{figure}[b!]\caption{a graphical representation of $\mathcal{T}'$ and the transitions to make it complete}\label{fig:Tp}
          \begin{center}
            \resizebox{\linewidth}{!}{
              \begin{tikzpicture}[auto, shorten >=1pt, >=latex]
                \node[state] (c) {$c$};
                \node[state, right=of c] (0) {$0$};
                \node[state, right=of 0] (1) {$1$};
                \node[state, right=of 1] (2) {$2$};
                \node[state, draw=none, right=of 2] (d) {$\dots$};
                \node[state, right=0cm of d] (r) {$r$};
                \node[state, right=of r] (r+1) {};
                \node[state, right=of r+1] (r+2) {};
                \node[state, below=1.25cm of 1] (1p) {$1'$};
                \node[state, right=of 1p] (2p) {$2'$};
                \node[state, draw=none, right=of 2p] (dp) {$\dots$};
                \node[state, right=0cm of dp] (rp) {$r'$};
                \node[state, right=of rp] (r+1p) {};
                \node[state, right=of r+1p] (r+2p) {};
                \node[state, below=3.75cm of c] (z) {$d$};
                \node[state, right=of z] (0pp) {$0''$};
                \node[state, right=of 0pp] (1pp) {$1''$};
                \node[state, right=of 1pp] (2pp) {$2''$};
                \node[state, draw=none, right=of 2pp] (dpp) {$\dots$};
                \node[state, right=0cm of dpp] (rpp) {$r''$};
                \node[state, right=of rpp] (r+1pp) {};
                \node[state, right=of r+1pp] (r+2pp) {};

                \path[->] (c) edge[loop below] node[align=left] {$0/0$\\$1/1$} (c)
                          (c) edge node[below] {$\#/\#$} (0)
                          (0) edge node[below, pos=0.6] {$1/1$} (1)
                          (0) edge[bend left=15] node[below left, pos=0.2] {$0/0$} (1p)
                          (1) edge node[below, pos=0.6] {$1/1$} (2)
                          (1) edge[bend left=15] node[below left, pos=0.2] {$0/0$} (2p)
                          (1p) edge node[below, align=left] {$0/0$\\$1/1$} (2p)
                          (2) edge node[below, pos=0.6] {$1/1$} (d)
                          (2) edge[bend left=15] node[below left, pos=0.2] {$0/0$} (dp)
                          (2p) edge node[below, align=left] {$0/0$\\$1/1$} (dp)
                          (r) edge node[below] {$\#/\#$} (r+1)
                          (rp) edge node[below] {$\#/\#$} (r+1p)
                          (r+1) edge node[below] {$1/1$} (r+2)
                          (r+1p) edge node[below] {$1/0$} (r+2p)
                          (z) edge[loop above] node[align=left] {$0/0$\\$1/1$} (z)
                          (z) edge node[below] {$\#/\#$} (0pp)
                          (0pp) edge node[align=left, below] {$1/1$\\$0/0$} (1pp)
                          (1pp) edge node[align=left, below] {$1/1$\\$0/0$} (2pp)
                          (2pp) edge node[align=left, below] {$1/1$\\$0/0$} (dpp)
                          (rpp) edge node[below] {$\#/\#$} (r+1pp)
                          (r+1pp) edge node[below] {$1/0$} (r+2pp)

                          (0) edge[bend left, dashed] node[above] {$\#/\#$} (1)
                          (1) edge[bend left, dashed] node[above] {$\#/\#$} (2)
                          (2) edge[bend left, dashed] node[above] {$\#/\#$} (d)
                          (r) edge[bend left, dashed] node[above, align=center] {$0/0$\\$1/1$} (r+1)
                          (r+1) edge[bend left, dashed] node[above, align=center] {$0/0$\\$\#/\#$} (r+2)
                          (r+2) edge[loop right, dashed] node[align=center] {$0/0$\\$1/1$\\$\#/\#$} (r+2)

                          (1p) edge[bend left, dashed] node[above, pos=0.3] {$\#/\#$} (2p)
                          (2p) edge[bend left, dashed] node[above, pos=0.3] {$\#/\#$} (dp)
                          (rp) edge[bend left, dashed] node[above, align=center] {$0/0$\\$1/1$} (r+1p)
                          (r+1p) edge[bend left, dashed] node[above, align=center] {$0/1$\\$\#/\#$} (r+2p)
                          (r+2p) edge[loop right, dashed] node[align=center] {$0/0$\\$1/1$\\$\#/\#$} (r+2p)

                          (0pp) edge[bend left, dashed] node[above] {$\#/\#$} (1pp)
                          (1pp) edge[bend left, dashed] node[above] {$\#/\#$} (2pp)
                          (2pp) edge[bend left, dashed] node[above] {$\#/\#$} (dpp)
                          (rpp) edge[bend left, dashed] node[above, align=center] {$0/0$\\$1/1$} (r+1pp)
                          (r+1pp) edge[bend left, dashed] node[above, align=center] {$0/1$\\$\#/\#$} (r+2pp)
                          (r+2pp) edge[loop right, dashed] node[align=center] {$0/0$\\$1/1$\\$\#/\#$} (r+2pp)
                ;
              \end{tikzpicture}}
          \end{center}
        \end{figure}
        Most interesting are the states $c$ (\enquote{check}) and $d$ (\enquote{disable}). On a word of the form $\{ 0, 1 \}^* \# \{ 0, 1 \}^r \# 1$, $c$ checks whether any of the letters in the $\{ 0, 1 \}^r$ block is $0$. In that case, it switches the trailing $1$ into a $0$; otherwise, it does not change the word at all. The behavior of $d$ is similar. However, it \emph{always} changes the trailing $1$ into a $0$.

        Combining these automata, we define $\mathcal{T}$ as the disjoint union of $\mathcal{T}'$,\allowbreak $\mathcal{T}_1,\allowbreak \mathcal{T}_2,\allowbreak \dots,\allowbreak \mathcal{T}_k$, $\mathcal{T}_1',\allowbreak \mathcal{T}_2',\allowbreak \dots,\allowbreak \mathcal{T}_k'$ and $\mathcal{T}_1'',\allowbreak \mathcal{T}_2'',\allowbreak \dots,\allowbreak \mathcal{T}_k''$ where we add, for each $k \in \{ 1, 2, \dots, r \}$, transitions with input $\#$ and output $\#$ from the non-final states $z \in Z_k \setminus F_k$ to the $0$-labeled state of $\mathcal{T}_k'$ (in $\mathcal{T}$) and from the final states $z \in F_k$ to the $0$-labeled state of $\mathcal{T}_k''$. We give a schematic representation of $\mathcal{T}$ in \autoref{fig:T}.
        \begin{figure}[h]\caption{Schematic representation of $\mathcal{T}$}\label{fig:T}
          \begin{center}
            \begin{tikzpicture}[auto, shorten >=1pt, >=latex]
            \node[rectangle, draw, minimum width=2.5cm, minimum height=4cm] (TkRect) {};
            \node[right=0.2cm of TkRect.west, anchor=north, rotate=90] (TkAk) {$\mathcal{T}_k$ (from $\mathcal{A}_k$)};
            \node[state, below left=0.75cm of TkRect.north east] (non-final) {};
            \node[state, accepting, above left=0.75cm of TkRect.south east] (final) {};
            \coordinate (topright) at ($(1cm, 0) + (TkRect.north east)$);
            \coordinate (bottomright) at ($(1cm, 0) + (TkRect.south east)$);
            \coordinate (topleft) at ($(-0.5cm, 0) + (TkRect.north west)$);
            \node[rectangle, draw, minimum width=5cm, minimum height=1.75cm, below right=0cm of topright] (TkpRect) {};
            \node[left=0.2cm of TkpRect.east, anchor=north, rotate=-90] (Tkp) {$\mathcal{T}_k'$};
            \node[state, below right=0.75cm of TkpRect.north west] (0p) {$0$};
            \path[->] (non-final) edge node {$\#/\#$} (0p);
            \node[rectangle, draw, minimum width=5cm, minimum height=1.75cm, above right=0cm of bottomright] (TkppRect) {};
            \node[left=0.2cm of TkppRect.east, anchor=north, rotate=-90] (Tkpp) {$\mathcal{T}_k''$};
            \node[state, above right=0.75cm of TkppRect.south west] (0pp) {$0$};
            \path[->] (final) edge node {$\#/\#$} (0pp);
            \node[rectangle, draw, minimum size=2cm, below left=0cm of topleft] {$\mathcal{T}'$};
            \end{tikzpicture}%
          \end{center}%
        \end{figure}
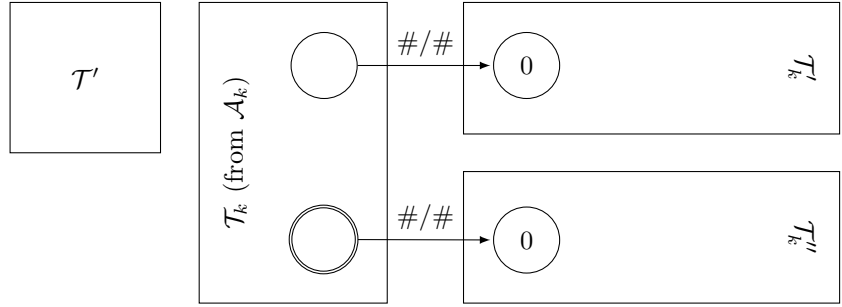
        
        We are going to use $\mathcal{T}$ as the automaton in the instance of the uniform word problem for automaton-inverse semigroups. Therefore, one should verify that $\mathcal{T}$ is an $\inverse{\mathscr{S}}$-automaton and can be computed from $\mathcal{A}_1, \mathcal{A}_2, \dots, \mathcal{A}_r$ in logarithmic space. Besides $\mathcal{T}$, we also need two finite sequences of its states to form a proper problem instance. As the sequence on the left hand size, we use $c z_{0, r} z_{0, r - 1} \dots z_{0, 1}$ and, on the right hand size, we use $d z_{0, r} z_{0, r - 1} \dots z_{0, 1}$.

        It remains to show that we have $c z_{0, r} z_{0, r - 1} \dots z_{0, 1} \circ u = d z_{0, r} z_{0, r - 1} \dots z_{0, 1} \circ u$ for all $u \in \{ 0, 1, \# \}^*$ if and only if $\bigcap_{k = 1}^r L(\mathcal{A}_k) = \emptyset$. For this, we first investigate the behavior on a word $u$ that is not a prefix of a word from $\{ 0, 1 \}^* \# 1^r \# 1$. By construction, we have that $z_{0, r} z_{0, r - 1} \dots z_{0, 1} \circ u$ is undefined because either $u$ is not even a prefix of a word from $\{ 0, 1 \}^* \# \{ 0, 1 \}^r \# 1$ -- in which case $z_{0, 1} \circ u$ is already undefined -- or it has a prefix in $\{ 0, 1 \}^* \# 1^s 0$ for an $s \in \mathbb{N}$. Let this prefix be $w \# 1^s 0$. Then, we have $z_{0, s} z_{0, s - 1} \dots z_{0, 1} \circ u = u' \in w \# \{ 0, 1 \}^s 0$. But on such a word, $z_{0, s + 1} \circ{}$ is undefined.

        Next, let $\bigcap_{k = 1}^r L(\mathcal{A}_k) = \emptyset$. By the previous discussion and due to the behavior of automata, we only have to show $c z_{0, r} z_{0, r - 1} \dots z_{0, 1} \circ u = d z_{0, r} z_{0, r - 1} \dots z_{0, 1} \circ u$ for all words from $\{ 0, 1 \}^* \# 1^r \# 1$. Let $u = w \# 1^r \# 1$ be such a word. Let $k_0$ be the smallest number such that $w \not\in L(\mathcal{A}_{k_0})$. By construction of $\mathcal{T}$, we have $z_k \circ u = u$ for $1 \leq k < k_0$: no automaton $\mathcal{T}_k$ changes the $w$ part of $u$ and, by the choice of $k_0$, we continue at the $0$-labeled state of $\mathcal{T}_k''$ after reading the first $\#$ in $u$; this does not change the $1^r \# 1$ part of $u$ either. Thus, we have $z_{0, k_0 - 1} z_{0, k_0 - 2} \dots z_{0, 1} \circ u = u$. Now, applying $z_{0, k_0} \circ u$, again, leaves the $w$ part unchanged (due to the construction of $\mathcal{T}_{k_0}$) but switches the $k^\textnormal{th}$ $1$ in the $\# 1^r \# 1$ part to $0$ (because, after reading $w \#$, we enter $\mathcal{T}_k'$ in $\mathcal{T}$). In the end, we have $z_{0, r} z_{0, r - 1} \dots z_{0, 1} \circ u = u' \in w \# 1^{k_0 - 1} 0 \{ 0, 1 \}^{r - k_0} \# 1$. Now, both, $c \circ{}$ and $d \circ{}$, only switch the trailing $1$ in $u'$ into a zero. Thus, we have $c \circ u' = d \circ u'$.

        Finally, let $w \in \bigcap_{k = 1}^r L(\mathcal{A}_k)$. Define $u = w \# 1^r \# 1$. For each $k \in \{ 1, 2, \dots, r \}$, we have $z_{0, k} \circ u = u$ because we are in a final state after reading $w$ in $\mathcal{A}_k$ and, thus, continue in $\mathcal{T}_k''$ which does not change $u$. Therefore, we have $z_{0, r} z_{0, r - 1} \dots\allowbreak z_{0, 1} \circ u = u$. Because there is no $0$ in the $1^r$ block of $u$, $c \circ{}$ will not change the trailing $1$ into a $0$ but $d \circ{}$ will, which shows $c z_{0, r} z_{0, r - 1} \dots z_{0, 1} \circ u \neq d z_{0, r} z_{0, r - 1} \dots z_{0, 1} \circ u$ and concludes the proof.

        We have shown the first two items of the proposition's assertion. For the uniform word problem for automaton \emph{groups}, we have to turn $\mathcal{T}$ into a complete automaton while keeping it deterministic and reverse deterministic. It is already complete in its $\mathcal{T}_k$ parts. In the $\mathcal{T}'$ and $\mathcal{T}_k'$ parts, we include the dashed transitions from \autoref{fig:Tp} and \autoref{fig:Tkp}. We may do the same for $\mathcal{T}_k''$; here, however, we have to change the (dashed) transition $\trans{(k - 1)}{0}{1}{k}$ to $\trans{(k - 1)}{0}{0}{k}$. One should verify that this new version of $\mathcal{T}$ is a $\mathscr{G}$-automaton. We only want to show $\PSPACE$-hardness of the uniform word problem for automaton groups \emph{with a single rational constraints}. The rational constraint we are going to use is that we only consider words from $\{ 0, 1 \}^* \# 1^r \# 1$. We may construct an according acceptor similar to the $d$ branch in \autoref{fig:Tp} in logarithmic space. This allows us to use the same proof as for the case of automaton-inverse semigroups above because the constraint ensures that we never use any of the dashed transitions (the new version of $\mathcal{T}$ behaves like the old version on all relevant words).
      \end{proof}

      The last proof shows $\PSPACE$-hardness of the stated uniform word problems even if one restricts the automata to the fixed ternary alphabet $\{ 0, 1, \# \}$. We want to remark that the result also holds for the stricter restriction to the binary alphabet $\{ a, b \}$. This can be shown by encoding $0$ as $aa$, $1$ as $ab$ and $\#$ as $bb$ and adapting the previous automata accordingly. It is important to note that, using this encoding, we preserve inverse-determinism; in particular, it is still possible to toggle $0$/$aa$ and $1$/$ab$ without reading and outputting two letters in one go.

      Next, we show a generalization of the previous proposition: we make the transition from the uniform to the non-uniform problems. This generalization, however, comes at the cost of increased complexity in the encoding techniques and the proof as a whole.
      \begin{prop}\label{prop:nonuniformPSPACEhardness}
        The following holds:
        \begin{itemize}
          \item There is an automaton semigroup for which the word problem is $\PSPACE$-hard.
          \item There is an automaton-inverse semigroup for which the word problem is $\PSPACE$-hard.
          \item There is an automaton group for which the word problem with a single rational constraint is $\PSPACE$-hard (where the rational constraint is part of the input).
        \end{itemize}
      \end{prop}
      \begin{proof}
        The first assertion is an implication of the second one. Therefore, we start with the inverse semigroup case and show the group case later.

        \paragraph{The Automaton-Inverse Semigroup Case}
        There exists a (deterministic) one-tape, $\PSPACE$-universal Turing Machine $\mathcal{M}_0$, i.\,e.\ a Turing Machine that halts on input of a word $w$ and the encoding of another $\PSPACE$ machine $\mathcal{M}$ if and only if $\mathcal{M}$ halts on $w$; furthermore, the space used by $\mathcal{M}_0$ on input of length $n$ is bounded by a polynomial $p(n)$. Let $\Gamma$ be the tape alphabet of $\mathcal{M}_0$, $Z$ its state set, $z_0$ its initial state and $\blank$ its blank symbol. Without loss of generality, we may assume that $\mathcal{M}_0$ never visits a tape position to the left of its initial position, which we call position $0$. Of course, $\mathcal{M}_0$ neither visits a position to the right of position $p(n) - 1$ (on input of length $n$). Equally without loss of generality, we may assume that the computation of $\mathcal{M}_0$ never stops (even in the case of acceptance). The machine accepts an input word if the corresponding (infinite) computation passes a final state; let $E$ be the set of these final states.

        Suppose, at time step $t$, $\mathcal{M}_0$'s tape contains the symbols $\gamma_0 \gamma_1 \dots \gamma_{p(n) - 1}$, its head is at position $i$ and it is in state $z$, then, we write its configuration as $\gamma_0 \gamma_1 \dots \gamma_{i - 1} (\gamma_i, z) \gamma_{i + 1} \gamma_{i + 2}\allowbreak \dots \gamma_{p(n) - 1}$, which we consider to be a word of length $p(n)$ over the alphabet $\Delta = \Gamma \cup (\Gamma \times Z)$. For example, on input $a_0 a_1 \dots a_{n - 1}$ (where $a_0, a_1, \dots, a_{n - 1}$ are single letters), the initial configuration, i.\,e.\ that at time $0$, of $\mathcal{M}_0$ is written as $(a_0, z_0) a_1 a_2 \dots a_{n - 1} \blank \blank \dots \blank$.

        Note that the letter at position $i$ in the configuration for time step $t$ is uniquely determined by the letters at positions $i - 1$, $i$ and $i + 1$ of the configuration at time step $t - 1$. Here, the symbols at position $-1$ (the one left to the initial position) and at position $p(n)$ are implicitly considered to be blanks. This yields the transition map $\tau: \Delta^3 \to \Delta$, which only depends on $\mathcal{M}_0$. The general idea of this proof is to create an automaton with a \enquote{checker} state which encodes the finite map $\tau$ to check a sequence of configurations $u = c_1 \# c_2 \# \dots$ separated by the new symbol $\# \not\in \Delta$ for valid transitions in each position. Let $c_t = \gamma_0^{(t)} \gamma_1^{(t)} \dots \gamma_{p(n) - 1}^{(t)}$ with $\gamma_i^{(t)} \in \Delta$ for all $0 \leq i < p(n)$ and $t > 0$. The first application of the checker state to $u$ will check all zeroth positions, i.\,e.\ $\gamma_0^{(1)}, \gamma_0^{(2)}, \dots$. If the transition is valid, it will put a check-mark on the corresponding letter in $u$. When applying the checker state for the second time, it will ignore all zeroth positions since they already have a check-mark and check all first positions; see \autoref{fig:checkmarking}.
        If we repeat this $p(n)$ times, we know whether the sequence of configurations is a valid computation of $\mathcal{M}_0$. This approach has a drawback, however: putting a check-mark on the next letter that, so-far, does not have one is an operation which results in a non-inverse-deterministic automaton. So, we cannot use this approach for automaton-inverse semigroups and automaton groups directly.

        \begin{figure}[h]\caption{Illustration of the check-mark approach.}\label{fig:checkmarking}%
          \centering
          \begin{tikzpicture}
            \matrix (m) [matrix of math nodes, every node/.style/.append={inner sep=0pt}] {
              \underset{\phantom{\checkmark}}{\gamma_0^{(1)}} & \gamma_1^{(1)} & \gamma_2^{(1)} & \dots & \gamma_{p(n) - 1}^{(1)} & \# & \gamma_0^{(2)} & \gamma_1^{(2)} & \gamma_2^{(2)} & \dots & \gamma_{p(n) - 1}^{(2)} & \# & \dots \\
              \underset{\checkmark}{\gamma_0^{(1)}} & \gamma_1^{(1)} & \gamma_2^{(1)} & \dots & \gamma_{p(n) - 1}^{(1)} & \# & \underset{\checkmark}{\gamma_0^{(2)}} & \gamma_1^{(2)} & \gamma_2^{(2)} & \dots & \gamma_{p(n) - 1}^{(2)} & \# & \dots \\
              \underset{\checkmark}{\gamma_0^{(1)}} & \underset{\checkmark}{\gamma_1^{(1)}} & \gamma_2^{(1)} & \dots & \gamma_{p(n) - 1}^{(1)} & \# & \underset{\checkmark}{\gamma_0^{(2)}} & \underset{\checkmark}{\gamma_1^{(2)}} & \gamma_2^{(2)} & \dots & \gamma_{p(n) - 1}^{(2)} & \# & \dots \\
            };
            \path[->] (m-1-1.west) edge[bend right] node[left] {checker} (m-2-1.west)
                      (m-2-1.west) edge[bend right] node[left] {checker} (m-3-1.west);
          \end{tikzpicture}%
        \end{figure}
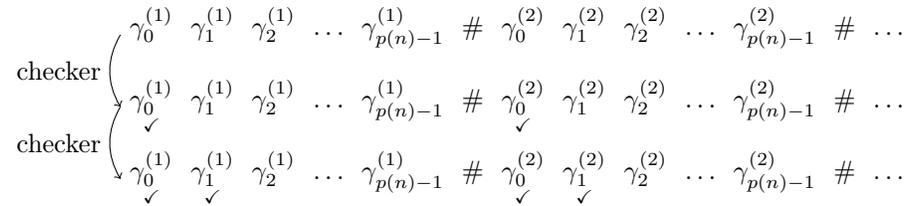%
        \paragraph*{Check-Marking using $\inverse{\mathscr{S}}$-Automata}
        To overcome this, we are going to use a modified version of the additive automaton from \autoref{expl:automatonForZ} and add a digit block after each configuration symbol. The construction will ensure that the block is always sufficiently long; this means a length of about $\log_2 p(n)$ but the exact length does not matter. For this block, we need two new symbols $0, 1 \not\in \Delta$. Instead of putting a check-mark on a symbol, we are going to do a $+1$ on its digit block. A symbol is considered to have a check-mark if and only if its block contains a $1$. For technical reasons, we will later also need to add a $\$ 00 \dots 0 \$ 0$ suffix to our configuration sequences,\footnote{Strictly speaking, this is only necessary for the group case.} where $\$$ is another new symbol, so we end up with the alphabet $\Sigma = \Delta \sqcup \{ 0, 1, \#, \$ \}$. The length of the block at the front of the suffix is again \enquote{long enough}. Summing this up, our automaton is designed to operate on words of the form
        \begin{align*}
          &\gamma_0^{(1)} 00 \dots 0 \quad \gamma_1^{(1)} 00 \dots 0 \quad \dots \quad \gamma_{p(n) - 1}^{(1)} \quad \# \\
          \textcolor{gray}{\drsh} \hspace{1cm}&\gamma_0^{(2)} 00 \dots 0 \quad \gamma_1^{(2)} 00 \dots 0 \quad \dots \quad \gamma_{p(n) - 1}^{(2)} \quad \# \quad \dots \quad \# \\
          \textcolor{gray}{\drsh} \hspace{1cm}&\gamma_0^{(T)} 00 \dots 0 \quad \gamma_1^{(T)} 00 \dots 0 \quad \dots \quad \gamma_{p(n) - 1}^{(T)} \quad \$ 00 \dots 0 \$ 0%
        \end{align*}
        where $T$ is some number and $\gamma_i^{(t)} \in \Delta$ for all $1 \leq t \leq T$ and all $0 \leq i < p(n)$.

        The automaton in \autoref{fig:PSpaceCheckmarking} -- more precisely, its $\checkmark$ state -- puts a check-mark in the sense described above on the first symbol without one in each configuration in the sequence. A transition label $\id_X$ means that, for every $x \in X$, there is such an $x/x$-labeled transition. Note that the automaton is deterministic and inverse-deterministic but not complete.
        \begin{figure}[h]%
          \centering
          \begin{tikzpicture}[auto, shorten >=1pt, >=latex]
            \node[state] (checkmark) {$\checkmark$};
            \node[state, right=of checkmark] (1) {};
            \node[state, above right=of 1] (2) {};
            \node[state, below right=of 2] (3) {};
            \node[state, below right=of 1] (4) {};
            \node[state, above left=of 2] (skip) {};
            \node[state, above right=of skip] (end) {};
            \node[right=2cm of 2, align=left] (2annotation) {\begin{varwidth}{5cm}\enquote{So far, the original input digit block of this symbol did not contain a $1$.}\end{varwidth}};
            \node[below right=of 3] (34annotation) {\begin{varwidth}{5cm}\enquote{The last symbol's digit block contained at least one $1$.}\end{varwidth}};
            \node[above right=of 2] (skipannotation) {\begin{varwidth}{5cm}\enquote{Skip everything up to the next configuration}\end{varwidth}};

            \path[->] (checkmark) edge node[below] {$\id_{\Delta}$} (1)
                      (1) edge node[below right] {$0/1$} (2)
                      (1) edge node[below left] {$1/0$} (4)
                      (2) edge[loop right, out=60, looseness=7, in=30] node {$0/0$} (2)
                      (2) edge node {$1/1$} (3)
                      (2) edge node[below left] {$\id_\Delta$} (skip)
                      (2) edge node[right] {$\$/\$$} (end)
                      (2) edge node[above] {$\#/\#$} (checkmark)
                      (3) edge[loop right] node[align=left] {$0/0$\\$1/1$} (3)
                      (3) edge node[below] {$\id_\Delta$} (1)
                      (4) edge[loop below] node {$1/0$} (4)
                      (4) edge node[below right] {$0/1$} (3)
                      (skip) edge[loop left] node {$\id_{\Delta \cup \{ 0, 1 \}}$} (skip)
                      (skip) edge node[above left] {$\#/\#$} (checkmark)
                      (skip) edge node {$\$/\$$} (end)
                      (end) edge[loop right] node {$\id_{\{ 0, 1, \$ \}}$} (end)
                      ;
            \path[->] (2annotation) edge[dotted] (2)
                      (34annotation) edge[dotted] (3)
                      (34annotation) edge[dotted] (4)
                      (skipannotation) edge[dotted] (skip);
          \end{tikzpicture}%
          \vspace*{-\baselineskip}
          \caption{The automaton used for generalized check-marking}\label{fig:PSpaceCheckmarking}%
        \end{figure}
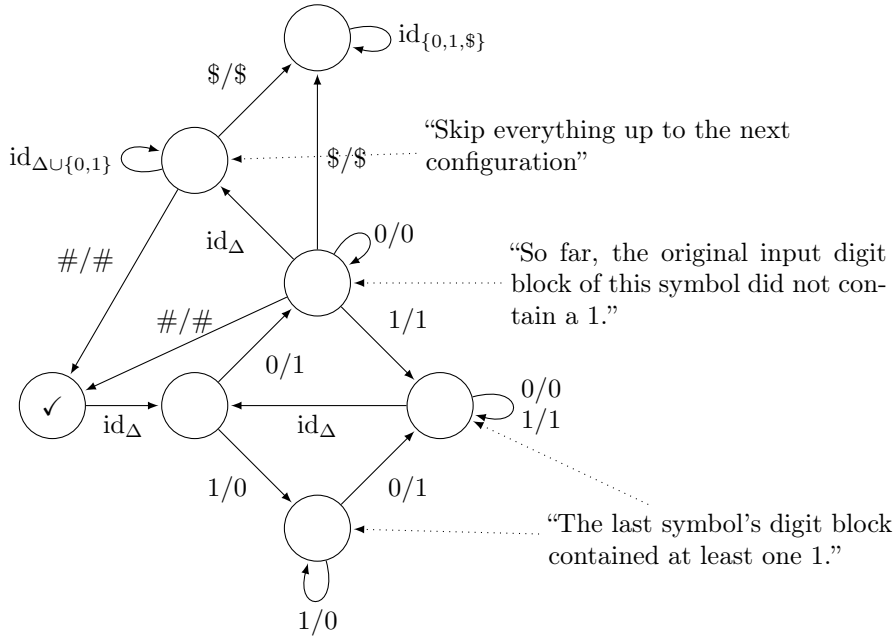%

        \paragraph{Checking the Validity}
        For the checking part of the automaton, we use the following states for each $\gamma_{-1}, \gamma_0, \gamma_1,\allowbreak \gamma'_{-1},\allowbreak \gamma'_0 \in \Delta$.
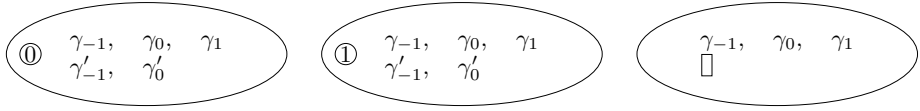
        \begin{center}
          \resizebox{\linewidth}{!}{%
            \begin{tikzpicture}[auto, shorten >=1pt, >=latex]
              \node[state, ellipse, align=center] (0) {\begin{tabular}{lll}
                 $\gamma_{-1}$,  & $\gamma_0$, & $\gamma_1$ \\
                 $\gamma'_{-1}$, & $\gamma'_0$    &
                \end{tabular}};
              \node[draw, circle, inner sep=0.1mm, right=2mm of 0.west] (s0) {$0$};

              \node[state, ellipse, align=center, right=0.5cm of 0] (1) {\begin{tabular}{lll}
                $\gamma_{-1}$,  & $\gamma_0$, & $\gamma_1$ \\
                $\gamma'_{-1}$, & $\gamma'_0$    &
                \end{tabular}};
              \node[draw, circle, inner sep=0.1mm, right=2mm of 1.west] (s1) {$1$};

              \node[state, ellipse, align=center, right=0.5cm of 1] (b) {\begin{tabular}{lll}
                $\gamma_{-1}$,  & $\gamma_0$, & $\gamma_1$ \\
                $\blank$ & &
                \end{tabular}};
            \end{tikzpicture}%
          }
        \end{center}
        The idea behind these states is the following. We have check-marked (in the more general sense mentioned above) the positions $0, 1, \dots, i - 1$ for all time steps in our configuration sequence, i.\,e.\ we now want to check the respective $i^\textnormal{th}$ position for all time steps. Suppose further that we already have successfully done this checking of position $i$ for the time steps up to and including $t - 1$. When we did the checking for time step $t - 1$, we stored the symbols at position $i - 1$, $i$ and $i + 1$ in the state. For storing this information, we use the upper row, i.\,e.\ the entires $\gamma_{-1}, \gamma_0, \gamma_1$. Now, we start reading the configuration (with the digit-blocks for each symbol) for time step $t$. During the reading, we remember the last two seen symbols of the configuration in the lower row. Whenever a new symbol begins, we have not seen its digit-block yet and, therefore, cannot know whether it is the first positions which has not been check-marked yet (i.\,e.\ position $i$). To handle this, we use the distinction between the circled \circledZero{} and circled \circledOne{} states. We are in the circled \circledZero{} state just after we have read $\gamma_0'$ and before we begin to read the corresponding digit-block. As long as we only see $0$ in this block, we stay in the \circledZero{} state. If we see a $1$, we know that the symbol has already been check-marked and we move to the \circledOne{} state where we skip everything up to the next symbol (updating $\gamma_{-1}'$ and $\gamma_{0}'$ accordingly). If we see the next symbol before we see a $1$, we know that $\gamma_0'$ is indeed the symbol at position $i$ of the time step $t$ and we can check the validity of the transition from step $t - 1$ to step $t$ at position $i$ using $\tau$. If the transition is valid, we want to skip the rest of the configuration at time step $t$ and begin again for time step $t + 1$. However, while we do the skipping, we have to store the symbols at positions $i - 1$, $i$ and $i + 1$ correctly for the next time step. This is what the last type of states (the ones with only a blank symbol in the lower row) is used for. We can also check the validity if we see one of the end markers $\#$ and $\$$ instead of the next symbol before we see a $1$. In the case of $\#$, we store the implicit $\blank$ as the symbol at position $i + 1$ and continue normally. In the case of $\$$, we skip to the end of the word.

        These ideas are implemented by the transitions which are schematically depicted in \autoref{fig:PSpaceCompleteTransitions}. We have these transitions for every $\gamma_{-1}, \gamma_0, \gamma_1,\allowbreak \gamma'_{-1},\allowbreak \gamma'_0 \in \Delta$.
        \begin{figure}[h]
          \begin{center}
            \resizebox{\linewidth}{!}{%
              \begin{tikzpicture}[auto, shorten >=1pt, >=latex]
                \node[state, ellipse, align=center] (0) {\begin{tabular}{lll}
                  $\gamma_{-1}$,  & $\gamma_0$, & $\gamma_1$ \\
                  $\gamma'_{-1}$, & $\gamma'_0$    &
                  \end{tabular}};
                \node[draw, circle, inner sep=0.1mm, right=2mm of 0.west] (s0) {$0$};

                \node[state, ellipse, align=center, below left=of 0] (1) {\begin{tabular}{lll}
                  $\gamma_{-1}$,  & $\gamma_0$, & $\gamma_1$ \\
                  $\gamma'_{-1}$, & $\gamma'_0$    &
                  \end{tabular}};
                \node[draw, circle, inner sep=0.1mm, right=2mm of 1.west] (s1) {$1$};

                \node[state, ellipse, align=center, below right=of 0] (b) {\begin{tabular}{lll}
                  $\gamma'_{-1}$,  & $\gamma'_0$, & $\gamma'_1$ \\
                  $\blank$ & &
                  \end{tabular}};

                \node[state, dotted, ellipse, align=center, below=of 1] (0') {\begin{tabular}{lll}
                  $\gamma_{-1}$, & $\gamma_0$, & $\gamma_1$ \\
                  $\gamma'_0$,   & $\gamma$    &
                  \end{tabular}};
                \node[draw, circle, inner sep=0.1mm, right=2mm of 0'.west] (s0') {$0$};

                \node[state, dotted, ellipse, align=center, below=of b] (1') {\begin{tabular}{lll}
                  $\gamma'_{-1}$, & $\gamma'_0$, & $\gamma'_1$ \\
                  $\blank$,       & $\blank$ &
                  \end{tabular}};
                \node[draw, circle, inner sep=0.1mm, right=2mm of 1'.west] (s1') {$1$};
                \node[left=0pt of 1'] (1'label) {$q_{\gamma'_{-1}, \gamma'_{0}, \gamma'_{1}} = {}$};

                \node[state, dotted, ellipse, align=center, above=of 0] (1b') {\begin{tabular}{lll}
                  $\gamma'_{-1}$, & $\gamma'_0$, & $\blank$ \\
                  $\blank$,       & $\blank$ &
                  \end{tabular}};
                \node[draw, circle, inner sep=0.1mm, right=2mm of 1b'.west] (s1b') {$1$};

                \node[state, above right=of 0] (d1) {};
                \node[state, right=of d1] (d2) {};
                \node[state, right=of d2] (d3) {};

                \path[->] (0) edge[loop left] node {$0/0$} (0)
                              edge node {$1/1$} (1)
                              edge[dashed] node[below left] {$\gamma'_1/\gamma'_1$} (b)
                              edge[dashed] node {$\#/\#$} (1b')
                              edge[dashed] node {$\$/\$$} (d1)
                          (1) edge[loop, out=120, in=150, looseness=7] node[above left, align=center] {$0/0$\\$1/1$} (1)
                              edge[dotted] node {$\gamma/\gamma$} (0')
                          (b) edge[loop, out=60, in=30, looseness=7] node[above] {$\id_{\Gamma \cup \Gamma \times Z \cup \{ 0, 1 \}}$} (b)
                              edge node {$\#/\#$} (1')
                              edge node {$\$/\$$} (d1)
                          (d1) edge[loop above] node[align=center] {$0/0$\\$1/1$} (d1)
                               edge node {$\$/\$$} (d2)
                          (d2) edge node[align=center] {$0/0$\\$1/1$} (d3)
                ;
              \end{tikzpicture}%
            }
          \end{center}
          \caption{Schematic representation of the transitions used for the checking part of the automaton and the definition of $q_{\gamma_{-1}, \gamma_{0}, \gamma_{1}}$}\label{fig:PSpaceCompleteTransitions}
        \end{figure}
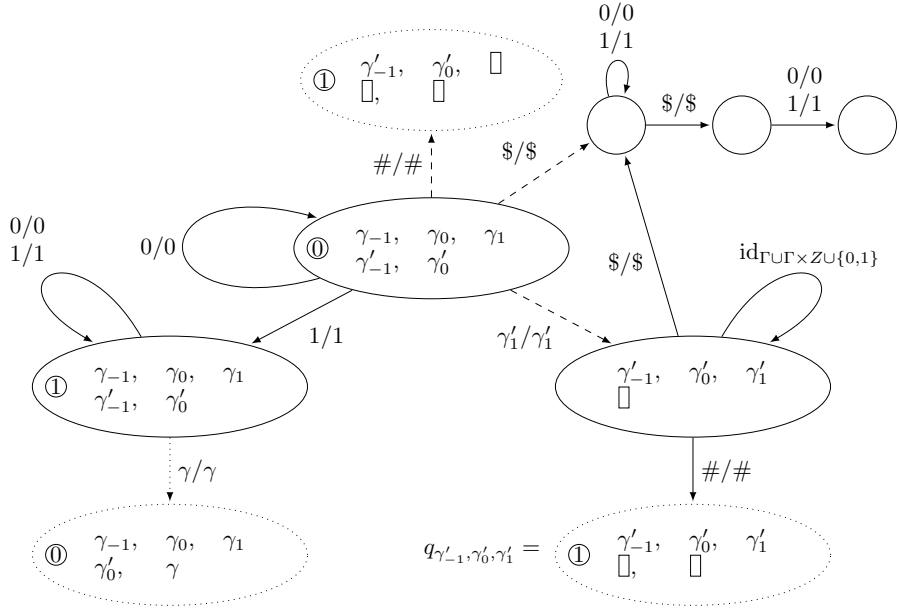
        The dashed transitions exist only if $\tau(\gamma_{-1}, \gamma_0, \gamma_1) = \gamma'_0$; if this is the case, then the transition labeled $\gamma'_1/\gamma'_1$ exists for every $\gamma'_1 \in \Delta$. The dotted $\gamma/\gamma$ transition exists for every $\gamma \in \Delta$ and the dotted states are equivalent to the non-dotted states for the corresponding other values of $\gamma_{-1}, \gamma_0, \gamma_1,\allowbreak \gamma'_{-1},\allowbreak \gamma'_0 \in \Delta$ (their outgoing edges are not drawn). Note that the functions given by the states always behave like the identity on values where they are defined and that the automaton is an $\inverse{\mathscr{S}}$-automaton.

        \paragraph{Behavior on Words Encoding Configuration Sequences}
        We will have a closer look on how the automaton behaves on a sequence of configurations next. Suppose $\mathcal{M}_0$ starts its computation in the initial configuration $\gamma_0^{(0)} \gamma_1^{(0)} \dots \gamma_{p(n) - 1}^{(0)} = (a_0, z_0) a_1 a_2 \dots a_{n - 1} \blank \blank \dots \blank$ on input $a_0 a_1 \dots a_{n - 1}$ (where $a_0, a_1, \dots a_{n - 1}$ are letters). Furthermore, let $c_t = \gamma_0^{(t)} \gamma_1^{(t)} \allowbreak\dots\allowbreak \gamma_{p(n) - 1}^{(t)}$ with $\gamma_i^{(t)} \in \Delta$ for all $0 \leq i < p(n)$ be the configuration at time $t > 0$ of this computation. Note that this configuration is well-defined for all $t \geq 0$ (by assumption on $\mathcal{M}_0$'s structure). At the border, we define $\gamma_{-1}^{(t)} = \gamma_{p(n)}^{(t)} = \blank$ for all $t \geq 0$. Finally, we define $c'_t = \gamma_0^{(t)} 0^k \gamma_1^{(t)} 0^k \dots \gamma_{p(n) - 1}^{(t)} 0^k$ as the result of adding the digit-blocks to $c_t$. Here, we set $k = \lceil \log_2 p(n) \rceil$, i.\,e.\ the digit-blocks allow us to count up to $p(n)$. Now, we consider the behavior of the automaton on the input word $u = c_1' \# c_2' \# \dots \# c_T' \$ 0^k \$ 0$ (for an arbitrary $T > 0$). Note that the word starts with the (modified) configuration for $t = 1$, not with the initial configuration ($t = 0$). Instead, the initial configuration is encoded in the element
        \[
          \boldsymbol{q}_{c_0} = \checkmark q_{\gamma_{p(n) - 2}^{(0)}, \gamma_{p(n) - 1}^{(0)}, \gamma_{p(n)}^{(0)}} \dots \checkmark q_{\gamma_{0}^{(0)}, \gamma_{1}^{(0)}, \gamma_{2}^{(0)}} \checkmark q_{\gamma_{-1}^{(0)}, \gamma_{0}^{(0)}, \gamma_{1}^{(0)}}
        \]
        (for the definition of $q_{\gamma_{-1}, \gamma_0, \gamma_1}$ please refer to \autoref{fig:PSpaceCompleteTransitions}). What happens if we apply $\boldsymbol{q}_{c_0} \circ{}$ to $u$? One may verify, that $q_{\gamma_{-1}^{(0)}, \gamma_{0}^{(0)}, \gamma_{1}^{(0)}} \circ u$ is defined and, thus, equal to $u$ because the configurations describe a valid computation. Applying $\checkmark \circ{}$ to $u$ next results in the word
        \begin{align*}
          u_1 ={}& \gamma_0^{(1)} \revbin_k(1) \enspace \gamma_1^{(1)} \revbin_k(0) \enspace \gamma_2^{(1)} \revbin_k(0) \enspace \dots \enspace \gamma_{p(n) - 1}^{(1)} \revbin_k(0) \#\\
          \textcolor{gray}{\drsh} \hspace{1cm}\phantom{u_1 ={}}& \gamma_0^{(2)} \revbin_k(1) \enspace \gamma_1^{(2)} \revbin_k(0) \enspace \gamma_2^{(2)} \revbin_k(0) \enspace \dots \enspace \gamma_{p(n) - 1}^{(2)} \revbin_k(0) \# \dots \# \\
          \textcolor{gray}{\drsh} \hspace{1cm}\phantom{u_1 ={}}& \gamma_0^{(T)} \revbin_k(1) \enspace \gamma_1^{(T)} \revbin_k(0) \enspace \gamma_2^{(T)} \revbin_k(0) \enspace \dots \enspace \gamma_{p(n) - 1}^{(T)} \revbin_k(0) \$ 0^k \$ 0 \text{,}
        \end{align*}
        where $\revbin_k(i)$ is the reverse binary representation of $i \in \mathbb{N}$ with length $k$. Applying $q_{\gamma_{0}^{(0)}, \gamma_{1}^{(0)}, \gamma_{2}^{(0)}} \circ{}$ to $u_1$ will check all $\gamma_1^{(t)}$ for validity. Again, the result will be $u_1$. After the next $\checkmark \circ{}$, we have
        \begin{align*}
          u_2 ={}& \gamma_0^{(1)} \revbin_k(2) \enspace \gamma_1^{(1)} \revbin_k(1) \enspace \gamma_2^{(1)} \revbin_k(0) \enspace \dots \enspace \gamma_{p(n) - 1}^{(1)} \revbin_k(0) \#\\
          \textcolor{gray}{\drsh} \hspace{1cm}\phantom{u_2 ={}}& \gamma_0^{(2)} \revbin_k(2) \enspace \gamma_1^{(2)} \revbin_k(1) \enspace \gamma_2^{(2)} \revbin_k(0) \enspace \dots \enspace \gamma_{p(n) - 1}^{(2)} \revbin_k(0) \# \dots \# \\
          \textcolor{gray}{\drsh} \hspace{1cm}\phantom{u_2 ={}}& \gamma_0^{(T)} \revbin_k(2) \enspace \gamma_1^{(T)} \revbin_k(1) \enspace \gamma_2^{(T)} \revbin_k(0) \enspace \dots \enspace \gamma_{p(n) - 1}^{(T)} \revbin_k(0) \$ 0^k \$ 0 \text{.}
        \end{align*}
        This continues and, in the end, we have
        \begin{align*}
          u_{p(n)} ={}&
            \gamma_0^{(1)} \revbin_k(p(n)) \enspace
            \gamma_1^{(1)} \revbin_k(p(n) - 1) \enspace
            \dots \enspace
            \gamma_{p(n) - 1}^{(1)} \revbin_k(1) \#\\
          \textcolor{gray}{\drsh} \hspace{1cm}\phantom{u_{p(n)} ={}}&
            \gamma_0^{(2)} \revbin_k(p(n)) \enspace
            \gamma_1^{(2)} \revbin_k(p(n) - 1) \enspace
            \dots \enspace
            \gamma_{p(n) - 1}^{(2)} \revbin_k(1) \# \dots \# \\
          \textcolor{gray}{\drsh} \hspace{1cm}\phantom{u_{p(n)} ={}}&
            \gamma_0^{(T)} \revbin_k(p(n)) \enspace
            \gamma_1^{(T)} \revbin_k(p(n) - 1) \enspace
            \dots \enspace
            \gamma_{p(n) - 1}^{(T)} \revbin_k(1) \$ 0^k \$ 0 \text{.}
        \end{align*}

        For comparison, we will now look at what happens if we start with a sequence of (valid) configurations which, however, is \emph{not} a valid computation of $\mathcal{M}_0$: assume, there is a (smallest) position $0 \leq i < p(n)$ such that there exits (a minimal) $t > 0$ for which the transition is not valid, i.\,e.\ we have $\tau(\gamma_{i - 1}^{(t - 1)}, \gamma_i^{(t - 1)}, \gamma_{i + 1}^{(t - 1)}) \neq \gamma_i^{(t)}$. If we apply
        \[
          \checkmark q_{\gamma_{i - 2}^{(0)}, \gamma_{i - 1}^{(0)}, \gamma_{i}^{(0)}} \dots \checkmark q_{\gamma_{0}^{(0)}, \gamma_{1}^{(0)}, \gamma_{2}^{(0)}} \checkmark q_{\gamma_{-1}^{(0)}, \gamma_{0}^{(0)}, \gamma_{1}^{(0)}} \circ{} \text{,}
        \]
        the behavior will be the same as for the case where the sequence belonged to a valid computation. Let $u_i$ be the resulting word. When we then apply $q_{\gamma_{i - 1}^{(0)}, \gamma_{i}^{(0)}, \gamma_{i + 1}^{(0)}} \circ{}$ to $u_i$, we see a difference: $q_{\gamma_{i - 1}^{(0)}, \gamma_{i}^{(0)}, \gamma_{i + 1}^{(0)}} \circ u_i$ is undefined because we miss the dashed transition from \autoref{fig:PSpaceCompleteTransitions} for time step $t$.

        \paragraph{Behavior on Other Words}
        This describes the behavior on words of the correct form but we also need to consider other words. There can be different reasons why this is the case. As a start, we only want to consider words that are prefixes of words in $\left( (\Delta 0^*)^* \# \right)^* (\Delta 0^*)^* \$ 0^* \$ 0$. For this, we use the automaton from \autoref{fig:PSpaceFormChecker} (which is an $\inverse{\mathscr{S}}$-automaton): $q_c \circ{}$ will be defined on a word if and only if the word is a prefix of one in the set.
        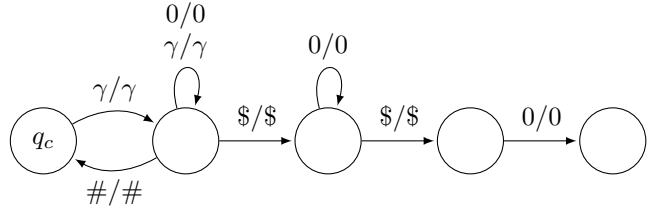
\begin{figure}[h]
          \begin{center}
            \begin{tikzpicture}[auto, shorten >=1pt, >=latex]
              \node[state] (qc) {$q_c$};
              \node[state, right=of qc] (1) {};
              \node[state, right=of 1] (2) {};
              \node[state, right=of 2] (3) {};
              \node[state, right=of 3] (4) {};

              \path[->] (qc) edge[bend left] node {$\gamma/\gamma$} (1)
                        (1) edge[bend left] node {$\#/\#$} (qc)
                            edge[loop above] node[align=center] {$0/0$\\$\gamma/\gamma$} (1)
                            edge node {$\$/\$$} (2)
                        (2) edge[loop above] node {$0/0$} (2)
                            edge node {$\$/\$$} (3)
                        (3) edge node {$0/0$} (4)
              ;
            \end{tikzpicture}
          \end{center}
          \caption{The automaton part used for checking the form of the input word. Transitions labeled with $\gamma/\gamma$ exist for every $\gamma \in \Delta$.}\label{fig:PSpaceFormChecker}
        \end{figure}

        Not every word on which $q_c \circ{}$ is defined is of the correct form, however. For example, the digit-block of a symbol can be too short so that we cannot count high enough (it cannot be too long because we ignore trailing $0$s). The good news is that this case is already handled: since there is no outgoing $\Delta$-transition from the state at the bottom of \autoref{fig:PSpaceCheckmarking}, $\boldsymbol{q}_{c_0} \circ{}$ will be undefined on a word that contains a digit-block which is too short because each $\checkmark$ will add one (in reverse/least significant bit first binary representation) to the digit block and we end up with a block consisting only of $1$s; the next application of $\checkmark$ will be undefined at the next letter after this block.\footnote{To see this, apply $\checkmark \circ{}\!$ \emph{twice} to $\gamma \revbin_2(2) \gamma \revbin_2(1) = \gamma 01 \gamma 10$ (for $\gamma \in \Delta$).} The same is true if the word encodes a sequence of configurations in which one configuration is shorter than $p(n)$. However, we also need to handle the case where one of the configurations is too long, i.\,e.\ longer than $p(n)$, because we do not check this part for validity. For this, we use $q_l$ from \autoref{fig:PSpaceCheckAllCheckmarked} (which is, again, an $\inverse{\mathscr{S}}$-automaton).
        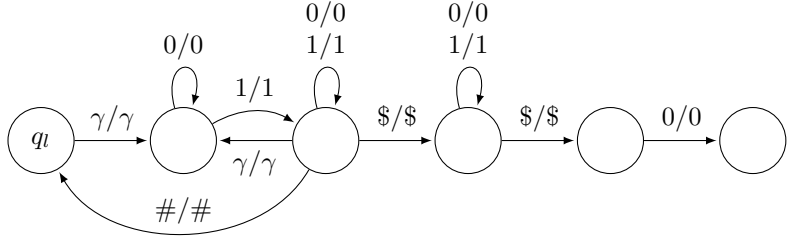
\begin{figure}[h]
          \begin{center}
            \begin{tikzpicture}[auto, shorten >=1pt, >=latex]
              \node[state] (ql) {$q_l$};
              \node[state, right=of ql] (1) {};
              \node[state, right=of 1] (2) {};
              \node[state, right=of 2] (3) {};
              \node[state, right=of 3] (4) {};
              \node[state, right=of 4] (5) {};

              \path[->] (ql) edge node {$\gamma/\gamma$} (1)
                        (1) edge[loop above] node {$0/0$} (1)
                            edge[bend left] node {$1/1$} (2)
                        (2) edge[loop above] node[align=center] {$0/0$\\$1/1$} (2)
                            edge node {$\gamma/\gamma$} (1)
                            edge[bend left=60] node[above] {$\#/\#$} (ql)
                            edge node {$\$/\$$} (3)
                        (3) edge[loop above] node[align=center] {$0/0$\\$1/1$} (3)
                            edge node {$\$/\$$} (4)
                        (4) edge node {$0/0$} (5)
              ;
            \end{tikzpicture}
            \caption{The automaton part used to check that, in each configuration, each symbol was check-marked. The transitions labeled with $\gamma/\gamma$ exist for every $\gamma \in \Delta$.}\label{fig:PSpaceCheckAllCheckmarked}
          \end{center}
        \end{figure}

        It checks that every digit-block contains at least one $1$ (i.\,e.\ that every symbol was check-marked). As there are exactly $p(n)$ many $\checkmark$s in $\boldsymbol{q}_{c_0}$, applying $\boldsymbol{q}_{c_0} \circ{}$ to a word (encoding a configuration sequence) will result in check-marks on the first $p(n)$ symbols of every configuration. If there is a symbol without a check-mark, the corresponding configuration is too long and $q_l \boldsymbol{q}_{c_0} \circ{}$ will be undefined on that word.

        Summing this up, we have that $q_l \boldsymbol{q}_{c_0} q_c \circ{}$ is defined on a word if and only if that word is a prefix of a word $w$ which satisfies the following conditions.
        \begin{itemize}
          \item $w$ is from $\left( (\Delta 0^*)^* \# \right)^* (\Delta 0^*)^* \$ 0^* \$ 0$, i.\,e.\ $w$ encodes a configuration sequence.
          \item The digit-block which belongs to a symbol at position $i$ (for any time step) allows for counting up to $p(n) - i$ in (reverse) binary.
          \item The encoded configurations are all of length exactly $p(n)$.
          \item $w$ encodes a valid computation of $\mathcal{M}_0$ when started in the initial configuration $c_0$.
        \end{itemize}

        \paragraph{Finding Final States}
        Thus, we have distinguished valid computations from other inputs. Next, we need to distinguish accepting computations (i.\,e.\ those where we reach a final state from $E$). For this, we use state $e$ from the $\inverse{\mathscr{S}}$-automaton depicted in \autoref{fig:PSpaceFinalStates}. This will toggle the trailing $0$ into a $1$ if and only if the word contains a final state and the digit block in the suffix $\$ 00 \dots 0 \$ 0$ contains only $0$.
        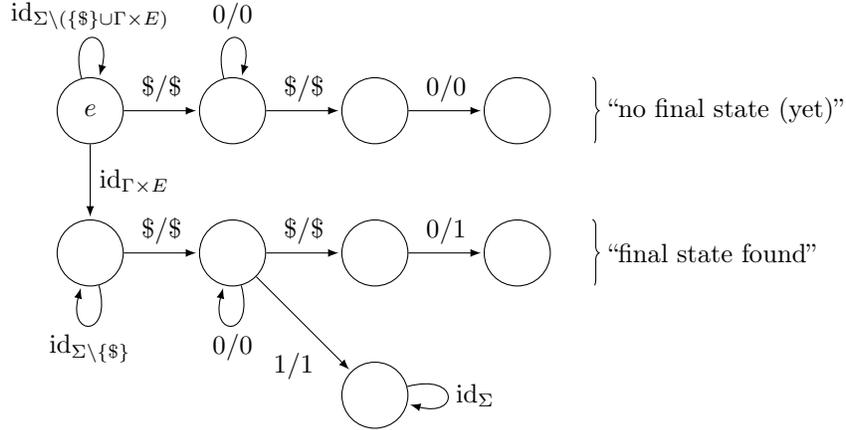
\begin{figure}[h]
          \begin{center}
            \begin{tikzpicture}[auto, shorten >=1pt, >=latex]
              \node[state] (e) {$e$};
              \node[state, right=of e] (1) {};
              \node[state, right=of 1] (2) {};
              \node[state, right=of 2] (3) {};
              \node[state, below=of e] (4) {};
              \node[state, right=of 4] (5) {};
              \node[state, right=of 5] (6) {};
              \node[state, right=of 6] (7) {};
              \node[state, below=of 6] (fail) {};

              \path[->] (e) edge[loop above] node {$\id_{\Sigma \setminus \left( \{ \$ \} \cup \Gamma \times E \right)}$} (e)
                            edge node {$\$/\$$} (1)
                            edge node {$\id_{\Gamma \times E}$} (4)
                        (1) edge[loop above] node {$0/0$} (1)
                            edge node {$\$/\$$} (2)
                        (2) edge node {$0/0$} (3)
                        (4) edge[loop below] node {$\id_{\Sigma \setminus \{ \$ \}}$} (4)
                            edge node {$\$/\$$} (5)
                        (5) edge[loop below] node {$0/0$} (5)
                            edge node {$\$/\$$} (6)
                            edge node[below left, pos=0.7] {$1/1$} (fail)
                        (6) edge node {$0/1$} (7)
                        (fail) edge[loop right] node {$\id_\Sigma$} (fail);
              ;
              \coordinate[right=1cm of 3.north] (as);
              \coordinate[right=1cm of 3.south] (ae);
              \draw[decorate, decoration=brace] (as) -- node[right=0.1cm] {\enquote{no final state (yet)}} (ae);
              \coordinate[right=1cm of 7.north] (bs);
              \coordinate[right=1cm of 7.south] (be);
              \draw[decorate, decoration=brace] (bs) -- node[right=0.1cm] {\enquote{final state found}} (be);
            \end{tikzpicture}
            \caption{The automaton part used to find final states in a valid computation.}\label{fig:PSpaceFinalStates}
          \end{center}
        \end{figure}

        \paragraph{The Reduction}
        This gives us all necessary parts to prove that the word problem for the automaton-inverse semigroup $S$ generated by the union of all the previous $\inverse{\mathscr{S}}$-automata is $\PSPACE$-hard. For this, we need to state the reduction function for any language $L \subseteq \Lambda^*$ which can be decided by a $\PSPACE$-machine $\mathcal{M}$. The function needs to map a problem instance of $L$, i.\,e.\ a word $w \in \Lambda^*$, to a problem instance of the word problem for $S$. Let $c_0(w)$ be the initial configuration of $\mathcal{M}_0$ on input of the encoding of $w$ and $\mathcal{M}$ (where the latter only depends on $L$). If $n$ denotes the length of the input encoding, then $c_0(w)$ has length $p(n)$. As $\mathcal{M}_0$ is a $\PSPACE$-universal Turing Machine, its computation starting in $c_0(w)$ accepts if and only $\mathcal{M}$ accepts $w$. Note that any position of $c_0(w)$ can be computed in logarithmic space. Thus, we can compute $e q_l \boldsymbol{q}_{c_0(w)} q_c$ and $q_l \boldsymbol{q}_{c_0(w)} q_c$ in logarithmic space as well. We use this pair of elements as the problem instance for the word problem (in fact, for the word problem's complement).

        Now, we show that the thus defined function is a reduction function. First, suppose that $w' \in L$. Then, $\mathcal{M}$ accepts $w'$ and $\mathcal{M}_0$ accepts the encoding $w$ of the pair $w'$ and $\mathcal{M}$. Let $n$ denote the length of this encoding. There is a sequence of configurations $c_0, c_1, \dots, c_T$ which describes a valid computation of $\mathcal{M}_0$ such that $c_0 = c_0(w)$ and $c_T$ contains a final state (i.\,e. a symbol from $\Gamma \times E$). As before, let $c_t'$ be the word obtained from $c_t$ by inserting $0^k$ blocks after each symbol with $k = \lceil \log_2 p(n) \rceil$. Then $v = q_l \boldsymbol{q}_{c_0(w)} q_c \circ c_1' \# c_2' \# \dots c_T' \$ 0^k \$ 0$ is defined and, because only the digits blocks have been manipulated (excluding the one in the $\$ 0^k \$ 0$ suffix), it contains an element from $\Gamma \times E$. Also, it ends with a $0$. This $0$ is toggled by $e \circ{}$ into a $1$. Thus, $e \circ v \neq v$ and $e q_l \boldsymbol{q}_{c_0(w)} q_c \neq q_l \boldsymbol{q}_{c_0(w)} q_c$ in the generated (inverse) semigroup.

        For the other direction, suppose $w' \not\in L$. Then, $\mathcal{M}_0$ will not accept the encoding $w$ of the pair $w'$ and $\mathcal{M}$. Now, if a word $u \in \Sigma^*$ is not a prefix of a valid encoding of $\mathcal{M}_0$'s computation on input $w$, then $q_l \boldsymbol{q}_{c_0(w)} q_c \circ{}$ will be undefined on $u$ and so will be $e q_l \boldsymbol{q}_{c_0(w)} q_c \circ{}$. If $u$ is a prefix of a valid encoding of the computation, then $u$ must not contain a symbol from $\Gamma \times E$ because $\mathcal{M}_0$ does not accept. Thus, $e$ will act like the identity on $q_l \boldsymbol{q}_{c_0(w)} q_c \circ u$. This yields $e q_l \boldsymbol{q}_{c_0(w)} q_c = q_l \boldsymbol{q}_{c_0(w)} q_c$ in the generated (inverse) semigroup, which concludes the proof for the automaton-inverse semigroup case.

        \paragraph{Extension to Groups}
        In the group case, we need a (single) rational constraint. This is used to ensure the well-formedness of the words on which the states act. The rational constraint we will be using depends on the input problem instance $w \in \Lambda^*$ to the reduction function. Let $n$ be the length of the encoding of the pair $w$ and $\mathcal{M}$ (one of the $\PSPACE$ machines for the problem to reduce) as input for $\mathcal{M}_0$ and let $k = \lceil \log_2 p(n) \rceil$. Then, the constraint is
        \[
          C(w) = \left( \left( \Delta 0^k \right)^{p(n)} \# \right)^* \left( \Delta 0^k \right)^{p(n)} \$ 0^k \$ 0 \text{.}
        \]
        Because we only need to count up to $p(n)$, we can construct an acceptor for this constraint in logarithmic space.

        Due to the constraint, we do not have to consider words that do not encode a sequence of (valid) configurations (we do have to consider words encoding a sequence of configurations that is not a valid computation of $\mathcal{M}_0$, however). So far, we have checked the validity of the transition from time step $t$ to time step $t + 1$ by including the dashed transition in \autoref{fig:PSpaceCompleteTransitions} only if $\tau(\gamma_{-1}, \gamma_0, \gamma_1) = \gamma_0'$. Since we need a complete automaton for the group case, we also need to have transitions if this is not the case. That is where the so-far unused $0$-block in the suffix $\$ 00 \dots 0 \$ 0$ comes into play; we use it in the $\inverse{\mathscr{S}}$-automaton depicted in \autoref{fig:PSpaceGroupFail}.
        \begin{figure}[h]
          \begin{center}
            \begin{tikzpicture}[auto, shorten >=1pt, >=latex]
              \node[state] (f) {$f$};
              \node[state, right=of f] (1) {};
              \node[state, right=of 1] (2) {};
              \node[state, right=of 2] (3) {};
              \node[state, right=of 3] (4) {};

              \path[->] (f) edge[loop above] node {$\id_{\Sigma \setminus \{ \$ \}}$} (f)
                        (f) edge node {$\$/\$$} (1)
                        (1) edge[loop above] node {$1/0$} (1)
                        (1) edge node {$0/1$} (2)
                        (2) edge[loop above] node[align=center] {$0/0$\\$1/1$} (2)
                        (2) edge node {$\$/\$$} (3)
                        (3) edge node {$0/0$} (4)
              ;
            \end{tikzpicture}
          \end{center}
          \caption{Automaton used for counting invalid transitions.}\label{fig:PSpaceGroupFail}
        \end{figure}
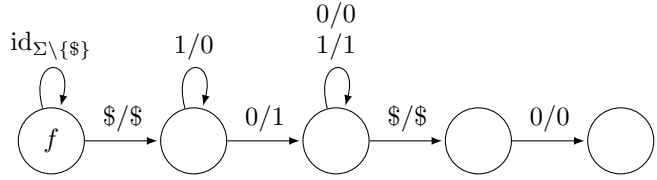
        The state $f$ will skip everything up to the first $\$$ and add one to the following digit block (which is considered to be a reverse/least significant bit first binary representation of some number). Remember that this will trigger a different behavior of $e$ later on. Now, if we do not have the dashed transitions in \autoref{fig:PSpaceCompleteTransitions}, we add transitions to $f$ instead. The behavior of the automaton may now intuitively be described as follows. If the configuration sequence does not belong to a valid computation of $\mathcal{M}_0$, then $\boldsymbol{q}_{c_0(w)} \circ{}\!$ is not undefined anymore but it will add one to the last digit block for every position with an invalid transition.

        In the other places where the automaton previously was not complete, we may add transitions into a sink state (in such a way that the resulting automaton is a $\mathscr{G}$-automaton). The rational constraint ensures that none of these transitions is actually ever used.

        Now, consider a problem $L$ that is decided by the $\PSPACE$-machine $\mathcal{M}$. First, suppose the problem instance $w \in \Lambda^*$ belongs to $L$. Then, the computation of $\mathcal{M}_0$ on the encoding of $w$ and $\mathcal{M}$ is accepting. Thus, there is a $u \in \Sigma^*$ which satisfies the constraint $C(w)$ and encodes this computation of $\mathcal{M}_0$ up to the first configuration that reaches a final state (i.\,e.\ contains an element from $\Gamma \times E$). On this $u$, we have $\boldsymbol{q}_{c_0(w)} \circ u \neq e \boldsymbol{q}_{c_0(w)} \circ u$ because the digit block at the end (counting the invalid transitions) remains $0^k$ and $e \circ{}$ toggles the trailing $0$ into a $1$.

        If $w \not\in L$, then in any word $u$ satisfying the constraint (i.\,e.\ encoding a sequence of valid configurations), there will either be an invalid transition or no final state is reached, i.\,e.\ $u$ does not contain a symbol from $\Gamma \times E$. In either case, $\boldsymbol{q}_{c_0(w)}$ and $e \boldsymbol{q}_{c_0(w)}$ will act on $u$ in the same way.

        This shows that mapping $w$ to the element pair $\boldsymbol{q}_{c_0(w)}$ and $e \boldsymbol{q}_{c_0(w)}$ and the rational constraint $C(w)$ is a reduction from $L$ to the (complement of the) word problem with a single rational constraint for the automaton group generated by the automaton described above. One may note that we use neither $q_c$ nor $q_l$ for this reduction since their roles are now covered by the rational constraint.
      \end{proof}

      So far, we have only been able to show $\PSPACE$-hardness of the word problem with a rational constraint in the group case. The question remains whether the problem is also $\PSPACE$-hard in the absence of rational constraints (in the uniform and non-uniform cases). Steinberg asks this question as well \cite[Question~5]{steinberg2015some} and suspects that the answer is \enquote{yes}. While the encoding techniques presented in the previous proof come close to an answer to this question (as they show $\PSPACE$-hardness for automaton-inverse semigroups), they still require some means of propagating a single bit from one element of the automaton structure to the next. It is not clear how this can be done if the input is not in some well-formed, well-known form or whether this can be done at all.

      However, we can show a weaker result, with which we conclude this section: $\NL$-hardness of the uniform case.
      \begin{prop}\label{prop:groupUniformNLhardness}
        The uniform word problem for automaton groups is $\NL$-hard.
      \end{prop}
      \begin{proof}
        It is easy to see that the following problem is $\NL$-hard.
        \problem
          [$\Sigma = \{ 0, 1 \}$]
          {a deterministic and complete $\Sigma$-acceptor $\mathcal{A}$}
          {is $L(\mathcal{A})$ empty?}
        For a $\Sigma$-acceptor $\mathcal{A} = (Z, \Sigma, \delta, z_0, F)$, we define the automaton $\mathcal{T}_{\mathcal{A}} = (Z, \Sigma, \delta')$ by setting
        \begin{align*}
          \delta' &= \{ \trans{z}{a}{a}{z'} \mid a \in \Sigma, z \in Z \setminus F, \transa{z}{a}{z'} \in \delta \} \cup{}\\
          &\phantom{{}={}} \{ \trans{z}{0}{1}{z'} \mid z \in F, \transa{z}{0}{z'} \in \delta \} \cup \{ \trans{z}{1}{0}{z'} \mid z \in F, \transa{z}{1}{z'} \in \delta \} \text{.}
        \end{align*}
        Structurally, $\mathcal{T}_{\mathcal{A}}$ behaves like $\mathcal{A}$. For non-final states $z \in Z \setminus F$, we have $z \circ 0 = 0$ and $z \circ 1 = 1$, while, for final states $z \in F$, $z \circ{}$ swaps $0$ and $1$. One should verify that, for all $\Sigma$-acceptors $\mathcal{A}$, $\mathcal{T}_{\mathcal{A}}$ is a $\mathscr{G}$-automaton by construction.

        We use $\mathcal{T}_{\mathcal{A}}$ to give a reduction from the problem stated above to the uniform word problem for automaton groups. An acceptor $\mathcal{A}$ is mapped to the automaton $\mathcal{T}_{\mathcal{A}}$ and the input sequence $z_0$, which consists of only a single state (i.\,e.\ $n = 1$). For the other sequence, we use $\idGrp$ (or, more precisely, the empty sequence with $m = 0$, which we consider to act as the identity on $\Sigma^*$). This mapping can, clearly, be computed by a deterministic Turing Machine in logarithmic space.

        Next, we show that $L(\mathcal{A}) = \emptyset$ if and only if $z_0 \circ{}$ is the identity on $\Sigma^*$. Suppose there is a word $u = a_1 a_2 \dots a_\ell \in L(\mathcal{A})$ (with $a_1, a_2, \dots, a_{\ell} \in \Sigma$). Then, we have $z_0 \cdot u = z_F$ for a $z_F \in F$. Therefore, we have $z_0 \circ u 0 = v (z_F \circ 0) = v 1 \neq u 0$ for a word $v \in \Sigma^*$ of length $\ell$. Thus, $z_0 \circ{}$ is not the identity on $\Sigma^*$. Now, suppose there is a word $u = a_1 a_2 \dots a_\ell$ (with $a_i \in \Sigma$) such that $u \neq v$ for $v = z_0 \circ u$. Since $z_0 \circ{}$ will map any prefix of $u$ to $v$'s prefix of the same length, we may, without loss of generality, assume that $u$ and $v$ differ only in the last letter. Let $z_F = z_0 \cdot a_1 a_2 \dots a_{\ell - 1}$. Then, we have $z_F \circ a_\ell \neq a_\ell$. However, this can only be the case if $z_F \in F$. Thus, we have $a_1 a_2 \dots a_{\ell - 1} \in L(\mathcal{A})$ and, therefore, $L(\mathcal{A}) \neq \emptyset$.
      \end{proof}
      If one bounds $n$ and $m$ in the uniform word problem for automaton groups by some constant, then, by \autoref{prop:uniformUpperBound}, this problem is in $\NSPACE(\log |\mathcal{T}|)$. By the Space Hierarchy Theorem \cite{Seiferas1973switching, Pap94}, it, therefore, cannot be $\PSPACE$-hard. This implies that, if one wants to show $\PSPACE$-hardness of the uniform word problem for automaton groups, then the sum $n + m$ must increase with the size of the instances. The same is in particular also true for showing $\PSPACE$-hardness of the word problem for a specific automaton group.
    \end{subsection}
  \end{section}

  \newpage
  \section*{Acknowledgements}
    The first author was supported by Austrian Science Fund (FWF) project P29355-N35. The second author acknowledges support from INDAM-GNSAGA.

  \section*{Summary of the contributions of this paper}
    \begin{itemize}
      \item We extend the notion of automaton semigroups to semigroups generated by partial automata. We show that, if $S$ is an automaton group (in the partial sense), then $S$ adjoined with a zero is a complete automaton semigroup (\autoref{prop:SAutSGImpliesS0CompAutSG}).
      \item We introduce the notion of automaton-inverse semigroups as semigroups generated by partial, invertible automata.
      \item We extend Steinberg's observation about an $\NSPACE(n)$ algorithm for the word problem of automaton groups to cover the uniform word problem for automaton semigroups, automaton-inverse semigroups and automaton groups, each with and without rational constraints (\autoref{prop:uniformUpperBound}). We also consider the algorithm for the non-uniform version of the word problem (\autoref{prop:nonuniformUpperBound}).
      \item We show that the upper bound on the length of a word on which two distinct elements of an automaton semigroup act differently must be exponential (\autoref{prop:lowerBoundOnWordLength}). The question about a better upper bound than the trivial one was raised by Cain \cite[Open problem 3.6]{cain2009automaton}.
      \item We give simple constructions which show that the uniform word problems for automaton semigroups and automaton-inverse semigroups are $\PSPACE$-hard and, thus, $\PSPACE$-complete (\autoref{prop:uniformPSPACEhard}). For automaton groups, the construction shows $\PSPACE$-hardness of the uniform word problem with a single rational constraint (thus, also, showing $\PSPACE$-completeness of this problem).
      \item We give a more complex construction which shows that there is an au\-to\-ma\-ton-inverse semigroup (and, thus, also an automaton semigroup) for which the word problem is $\PSPACE$-complete (\autoref{prop:nonuniformPSPACEhardness}). We extend the construction to show that there is an automaton group for which the word problem with a single rational constraint is $\PSPACE$-complete.
      \item We give a simple reduction to show $\NL$-hardness of the uniform word problem for automaton groups (\autoref{prop:groupUniformNLhardness}).
    \end{itemize}

\end{document}